\documentclass[11pt,a4paper]{article}
\usepackage{xspace}
\usepackage{vaucanson-g}
\usepackage{amsmath,amssymb,amsthm,amscd,mathrsfs}
\usepackage{clrscode} 

\newtheorem{theorem}{Theorem}
\newtheorem{proposition}[theorem]{Proposition}
\newtheorem{corollary}[theorem]{Corollary}
\newtheorem{lemma}[theorem]{Lemma}

\newcommand{\Sect}{\operatorname{Sect}}

\def\A{\mathcal{A}}
\def\B{\mathcal{B}} 
\def\C{\mathscr{C}} 
\def\R{{R}}

\def\T{{T}}

\DeclareMathOperator{\height}{\text{height}}

\newcommand{\ie}{{\itshape i.e.}\xspace }
\newcommand{\eg}{{\itshape e.g.}\xspace }
\newcommand{\etal}{{\itshape et al.}\xspace }
\newcommand{\resp}{{resp.}\xspace }
\ChgStateLineWidth{0.5}
\ChgEdgeLineWidth{0.5}
\FixVCScale{0.4}

\title{A quadratic algorithm for road coloring}
\author{Marie-Pierre B{\'e}al and Dominique Perrin
  \thanks{Universit\'e Paris-Est, Laboratoire d'informatique
    Gaspard-Monge CNRS UMR 8049, 5 boulevard Descartes, 77454
    Marne-la-Vall\'ee, France,
    \{beal,perrin\}@univ-mlv.fr}
    \thanks{This work is supported by
    French National Agency (ANR) through "Programme d'Investissements
    d'Avenir" (Project ACRONYME $\text{n}^\circ$ANR-10-LABX-58).}  }

\PlainState
\begin{document}
\maketitle

\begin{abstract}
  The Road Coloring Theorem states that every aperiodic directed graph
  with constant out-degree has a synchronized coloring. This theorem
  had been conjectured during many years as the Road Coloring Problem
  before being settled by A. Trahtman. Trahtman's proof leads to an
  algorithm that finds a synchronized labeling with a cubic worst-case
  time complexity.  We show a variant of his construction with a
  worst-case complexity which is quadratic in time and linear in
  space. We also extend the Road Coloring Theorem to the periodic
  case.
\end{abstract}
\section{Introduction}

Imagine a map with roads which are colored in such a way that a fixed
sequence of colors, called a homing sequence, leads the traveler to a
fixed place whatever the starting point is.  Such a coloring of the
roads is called synchronized and finding a synchronized coloring is
called the Road Coloring Problem. In terms of graphs, it consists in
finding a synchronized labeling in a directed graph.

The Road Coloring Theorem states that every aperiodic directed graph
with constant out-degree has a synchronized coloring (a graph is
aperiodic if it is strongly connected and the gcd of the length of the
cycles is equal to 1).  It has been conjectured under the name of the
Road Coloring Problem by Adler, Goodwin, and Weiss
\cite{AdlerGoodwynWeiss77}, and solved for many particular types of
automata (see for instance~\cite{AdlerGoodwynWeiss77},
\cite{OBrien81}, \cite{Carbone01}, \cite{Kari03}, \cite{Friedman90},
\cite{PerrinSchutzenberger92}). Trahtman settled the conjecture
in~\cite{Trahtman09}. In this paper, by Road Coloring Problem we
understand the algorithmic problem of finding a synchronized coloring
on a given graph (and not the existence of a polynomial algorithm
which is solved by the Road Coloring Theorem).

Solving the Road Coloring problem in each particular case is not only
a puzzle but has many applications in various areas like coding or
design of computational systems.  These systems are often modeled by
finite-state automata (\ie graphs with labels).  Due to some noise,
the system may take a wrong transition.  This noise may for instance
result from the physical properties of sensors, from unreliability of
computational hardware, or from insufficient speed of the computer
with respect to the arrival rate of input symbols.  It turns out that
the asymptotic behavior of synchronized automata is better than the
behavior of unsynchronized ones (see~\cite{DelyonMaler94}).
Synchronized automata are thus less sensitive to the effect of noise.

In the domain of coding, automata with outputs (\ie transducers) can
be used either as encoders or as decoders. When they are synchronized,
the behavior of the coder (or of the decoder) is improved in the
presence of noise or errors (see \cite{BerstelPerrinReutenauer2010},
\cite{Jurgensen2008}). For instance, the well-known Huffman
compression scheme leads to a synchronized decoder provided the
lengths of the codewords of the Huffman code are relatively prime.  It
is also a consequence of the Road Coloring Theorem that coding schemes
for constrained channels can have sliding block decoders and
synchronized encoders (see \cite{AdlerCoppersmithHassner83} and
\cite{LindMarcus95}).
 
Trahtman's proof is constructive and leads to an algorithm that finds
a synchronized labeling with a cubic worst-case time
complexity~\cite{Trahtman09,Trahtman2011}. The algorithm consists in
a sequence of flips of edges going out of some state so that the
resulting automaton is synchronized. One first searches a sequence of
flips leading to an automaton which has a so-called stable pair of
states (\ie with good synchronizing properties). One then
computes the quotient of the automaton by the congruence generated by
the stable pairs. The process is then iterated on this smaller
automaton. Trahtman's method for finding the sequence of flips leading
to a stable pair has a worst-case quadratic time complexity, which
makes his algorithm cubic.

In this paper, we design a worst-case linear time algorithm for
finding a sequence of flips until the automaton has a stable pair.
This makes the algorithm for computing a synchronized coloring
quadratic in time and linear in space.  The sequence of flips is
obtained by fixing a color, say red, and by considering the red cycles
formed with red edges, taking into account the positions of the roots
of red trees attached to each cycle. The prize to pay for decreasing
the time complexity is some more complication in the choice of the
flips. We also extend the Road Coloring Theorem to periodic graphs by
showing that Trahtman's algorithms provides a minimal-rank coloring.
Another proof of this result using semigroup tools, obtained
independently, is given in \cite{BudzbanFeinsilver11}. For related
results, see also \cite{Trahtman2010} and \cite{JonoskaKarl1999}.

The complexity of synchronization problems on automata has been
already studied (see \cite{KariVolkov2013} for a survey). It is
well-known that there is an $O(n^2)$ algorithm to test whether an
$n$-state automaton on a fixed-size alphabet is synchronized. The
complexity of computing a specific synchronizing word is $O(n^3)$ (see
\cite{Eppstein1990}).  However, the complexity of finding a
synchronizing word of a given length is NP-complete
\cite{Eppstein1990} (see also \cite{OlschewskiUmmels2010},
\cite{Roman2011}). The complexity of problems on automata has also
been studied for random automata (see
\cite{CarayolNicaud2012}). Several results prove that, under
appropriate hypotheses, a random irreducible automaton is synchronized
\cite{FreilingEtAl2003}, \cite{SkvortsovZaks2010}, and
\cite{Nicaud2013}. The average time complexity of these problems does not
seem to be known. In particular, we do not know the average time
complexity of the Road Coloring Problem.

The article is organized as follows. In Section~\ref{section.road}, we
give some definitions to formulate the problem in terms of finite
automata instead of graphs. In Section~\ref{section.algo} we describe
Trahtman's algorithm and our variant is detailed in
Section~\ref{section.algo2}. We give both an informal description of
the algorithm with pictures illustrating the constructions, and a
pseudocode. The time and space complexity of the algorithm are
analyzed in Section~\ref{section.complexity}. The periodic case is
treated in Section~\ref{section.periodic}. A preliminary version of this paper was posted in \cite{BealPerrin2008}.

\section{The Road Coloring Theorem} \label{section.road} 

In order to formulate the \emph{Road Coloring Problem} we introduce
the notation concerning automata.  

Let $A$ be a finite alphabet and let $Q$ be a finite set.  We denote
by $A^*$ the set of words over $A$. 

A (finite) \emph{automaton} $\A=(Q,E)$ over the alphabet $A$ with $Q$
as set of states is a given by a set $E$ of edges which are triples
$(p,a,q)$ where $p,q$ are states and $a$ is a symbol from $A$ called
the label of the edge. Note that no initial or final states are
specified.  Let $F$ be the multiset formed of the pairs $(p,q)$
obtained from the set $E$ by the map $(p,a,q) \mapsto (p,q)$.  The
multigraph having $Q$ as set of vertices and $F$ as set of edges is
called the \emph{underlying graph} of $\A$.

A \emph{path} in the automaton is sequence of consecutive edges. The
label of the path $((p_i,a_i,p_{i+1})_{0 \leq i \leq n}$ is the word $a_0
  \dotsm a_{n}$. The state $p_0$ is its origin and $p_{n+1}$ is its end.
The \emph{length} of the path is $n+1$. The path is a \emph{cycle} if $p_0=p_{n+1}$.

An automaton is \emph{deterministic} if, for each state $p$ and each
letter $a$, there is at most one edge starting at $p$ and labeled with
$a$.  It is \emph{complete deterministic} if, for each state $p$ and
each letter $a$, there is exactly one edge starting at $p$ and
labeled with $a$.  This implies that for each state $p$ and each word
$w$ there is exactly one path starting at $p$ and labeled with $w$.
The end of this unique path is denoted by $p \cdot w$.

An automaton is \emph{irreducible} if its underlying graph is strongly
connected. The \emph{period} of an automaton is the gcd of length of
its cycles. An automaton is \emph{aperiodic} if it is
irreducible and of period $1$\footnote{Note that this notion, which is
  usual for graphs, is not the notion of aperiodic automata used
  elsewhere and which refers to the period of words labeling the
  cycles (see \eg \cite{Eilenberg76B}).}.

A \emph{synchronizing word} of a complete deterministic automaton
$\A=(Q,E)$ is a word $w \in A^*$ such that for every pair of states
$p,q \in Q$, one has $p \cdot w = q \cdot w$.  A synchronizing word is
also called a \emph{reset sequence} \cite{Eppstein1990}, or a
\emph{magic sequence} \cite{BoyleMass2000,BoyleMass2004}, or also a
\emph{homing word} \cite{PomeranzReddy1994}. An automaton which has a
synchronizing word is called \emph{synchronized} (see an example on
the right part of Fig.~\ref{figure.automate1}).

Two automata which have isomorphic underlying graphs are called
\emph{equivalent}. Hence two equivalent automata differ only by the
labeling of their edges. In the sequel, we shall consider only complete
deterministic automata.

\begin{proposition} A synchronized complete deterministic automaton is aperiodic.
\end{proposition}
\begin{proof}
We assume that the automaton has at least one edge.  Let $(p,a,q)$ be
an edge of the automaton.  Let $w$ be a synchronizing word focusing to
a state $r$.  Since the graph is strongly connected, there is a word
$v$ such that from $r \cdot v = p$. Thus $p \cdot awv p = p \cdot
wv p$.  The lengths of the cycles from $p$ to $p$ labeled $awv$ and
$wv$ differ by $1$.  This implies that the period of automaton is $1$.
\end{proof}

The \emph{Road Coloring Theorem} can be stated as follows.

\begin{theorem}[A. Trahtman~\cite{Trahtman09}] \label{theorem.road}
Any aperiodic complete deterministic automaton is equivalent to a synchronized one.
\end{theorem}

\begin{figure}[htbp]
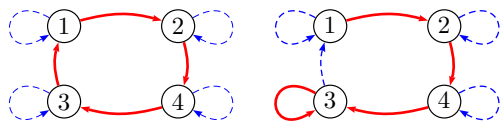

    \centering
\FixVCScale{0.5}
\VCDraw{%
\begin{VCPicture}{(0,-1)(12,2)}
\MediumState
\State[1]{(0,2)}{1}
\State[2]{(3,2)}{2}
\State[3]{(0,0)}{3}
\State[4]{(3,0)}{4}
\ChgEdgeLineColor{blue}
\ChgEdgeLineStyle{dashed}
\LoopW{1}{}
\LoopE{2}{}
\LoopW{3}{}
\LoopE{4}{}
\ChgEdgeLineWidth{2}
\ChgEdgeLineStyle{solid}
\ChgEdgeLineColor{red}
\ArcL[.5]{1}{2}{}
\ArcL[.5]{2}{4}{}
\ArcL[.5]{4}{3}{}
\ArcL[.5]{3}{1}{}
\VCPut{(7,0)}{
\MediumState
\State[1]{(0,2)}{1}
\State[2]{(3,2)}{2}
\State[3]{(0,0)}{3}
\State[4]{(3,0)}{4}
\ChgEdgeLineWidth{2}
\ChgEdgeLineColor{red}
\LoopW{3}{}
\ArcL[.5]{1}{2}{}
\ArcL[.5]{2}{4}{}
\ArcL[.5]{4}{3}{}
\ChgEdgeLineWidth{1}
\ChgEdgeLineStyle{dashed}
\ChgEdgeLineColor{blue}
\LoopW{1}{}
\LoopE{2}{}
\LoopE{4}{}
\ArcL[.5]{3}{1}{}
}
\end{VCPicture}%
        }
        \caption{Two complete aperiodic deterministic automata
          over the alphabet $A=\{a,b\}$. A thick red plain edge is an edge labeled by
          $a$ while a thin blue dashed edge is an edge labeled by
          $b$. The automaton on the left is not synchronized.  The one
          on the right is synchronized. For instance, the word $aaa$
          is a synchronizing word. The two automata are equivalent
          since their underlying graph are isomorphic.}\label{figure.automate1}
\end{figure}

A trivial case for solving the Road Coloring Theorem is the case where
the automaton has a loop edge around some state $r$
\cite{OBrien81}. Indeed, since the graph of the automaton is strongly
connected, there is a spanning tree rooted at~$r$ (with the edges of
the tree oriented towards the root). Let us label the edges of this
tree and the loop by the letter $a$. This coloring is synchronized by
the word $a^{n-1}$, where $n$ is the number of states.

\section{An algorithm for finding a synchronized coloring} \label{section.algo}

Trahtman's proof of Theorem~\ref{theorem.road} is constructive and
gives an algorithm for finding a labeling (also called a coloring) which
makes the automaton synchronized provided it is aperiodic.

In the sequel $\A$ denotes an $n$-state complete deterministic
automaton over an alphabet $A$. We fix a particular letter $a \in A$.
Edges labeled by $a$ are also called \emph{red edges} or
$a$-\emph{edges}. The other ones are called \emph{blue} or
$b$-\emph{edges}.

A pair $(p,q)$ of states in an automaton is \emph{synchronizable} if
there is a word $w$ with $p \cdot w = q \cdot w$.  It is well-known
that an automaton is synchronized if all its pairs of states are
synchronizable (see for instance Proposition~3.6.5 in
\cite{BerstelPerrinReutenauer2010}). 

A pair $(p,q)$ of states in an automaton is \emph{stable} if and only
if, for any word $u$, the pair $(p \cdot u, q \cdot u)$ is
synchronizable. This notion was introduced in \cite{CulikEtAl1999}.
In a synchronized automaton, any pair of states is stable.  Note that
if $(p,q)$ is a stable pair, then for any word $u$, $(p \cdot u, q
\cdot u)$ is also a stable pair, hence the terminology. Note also that
if $(p,q)$ and $(q,r)$ are stable pairs then $(p,r)$ is also a stable
pair. It follows that the relation defined on the set of states by $p
\equiv q$ if $(p,q)$ is a stable pair is an equivalence relation. As
observed in \cite[Lemma~2]{Kari03}, this relation is a congruence (\ie
$p \cdot u \equiv q \cdot u$ whenever $p \equiv q$) called \emph{the
  stable pair congruence}.  More generally, a congruence is
\emph{stable} if any pair of states in the same class is stable. The
congruence \emph{generated} by a stable pair $(p,q)$ is the least
congruence such that $p$ and $q$ belong to the same class. It is a
stable congruence. Given a congruence on the states of an automaton,
we denote by $\bar{p}$ the class of a state $p$.

If $\A=(Q,E)$ is an automaton, the \emph{quotient} of $\A$ by a stable
pair congruence is the automaton $\B$ whose states are the classes of
$Q$ under the congruence. The edges of $B$ are the triples
$(\bar{p},c,\bar{q})$ where $(p,c,q)$ is an edge of $\A$.  The
automaton $\B$ is complete deterministic when $\A$ is complete
deterministic. The automaton $\B$ is irreducible (\resp aperiodic)
when $\A$ is irreducible (\resp aperiodic).

The following Lemma was obtained by Culik \etal \cite{CulikEtAl02}.
We reproduce the proof since it helps understanding Trahtman's algorithm
(see the procedure \textsc{FindSynchronizedColoring} below).

\begin{lemma}[Culik \etal \cite{CulikEtAl02}] \label{lemma.lift} If
  the quotient of an automaton $\A$ by a stable pair congruence
  is equivalent to a synchronized automaton, then there is a
  synchronized automaton equivalent to $\A$.
\end{lemma}

\begin{proof}
  Let $\B$ be the quotient of $\A$ by a stable congruence and let
  $\B'$ be a synchronized automaton equivalent to $\B$.  We define an
  automaton $\A'$ equivalent to $\A$ as follows. The number of edges
  of $\A$ going out of $p$ and ending in states belonging to a same
  class $\bar{q}$ is equal to the number of edges of $\B$ (and thus
  $\B'$) going out of $\bar{p}$ and ending in $\bar{q}$. We define
  $\A'$ by labeling these edges according to the labeling of
  corresponding edges in $\B'$. The automaton $\B'$ is a quotient
  of~$\A'$.
 

  Let us show that $\A'$ is synchronized. Let $w$ be a synchronizing
  word of $\B'$ and $r$ the state ending any path labeled by $w$ in
  $\B'$. Let $p,q$ be two states of $\A'$.  Then $p \cdot w$ and $q
  \cdot w$ belong to the same congruence class. Hence $(p \cdot w, q
  \cdot w)$ is a stable pair of $\A'$. Therefore $(p,q)$ is a
  synchronizable pair of $\A'$. Since all pairs of $\A'$ are
  synchronizable, $\A'$ is synchronized.
\end{proof}

Trahtman's algorithm for finding a synchronized coloring of an
aperiodic automaton $\A$ consists in finding an equivalent automaton
$\A'$ of $\A$ which has at least one stable pair $(s,t)$, then in
recursively finding a synchronized coloring $\B'$ for the quotient
automaton $\B$ by the congruence generated by $(s,t)$, and finally in
lifting up this coloring to the initial automaton as follows. If there
is an edge $(p,c,q)$ in $\A$ but no edge $(\bar{p},c,\bar{q})$
in $\B'$, then there is an edge $(\bar{p},d,\bar{q})$ in $\B'$ with
$c\neq d$. Then we flip the labels of the two edges labeled $c$
and $d$ going out of $p$ in $\A$.

The algorithm for finding a synchronized coloring is described in the
following pseudocode.  The procedure \textsc{FindStablePair}, which
finds an equivalent automaton which has a stable pair of states, is
described in the next section. The procedure \textsc{Merge} computes
the quotient of an automaton by the stable congruence generated by a
stable pair of states.  The procedure \textsc{Update} updates some
data needed for the computation as described in
Section~\ref{section.update}.

\begin{small}
\begin{codebox}
\Procname{$\proc{FindSynchronizedColoring}(\text{aperiodic automaton } \A, \text{quotient automaton } \B)$} 
\li  $\B \gets \A$
\li  \While (size($\B$) $ > 1$) 
\li      \Do \proc{Update($\B$)}
\li       $\B, (s,t) \gets$ \proc{FindStablePair($\B$)}
\li       lift the coloring up from $\B$ to the automaton $\A$ 
\li       $\B \gets$ \proc{Merge$(\B,(s,t))$}
  \End
\li \Return $\A$ 
\end{codebox}
\end{small}


The termination of the algorithm is guaranteed by the fact that the
number of states of the quotient automaton of $\B$ is strictly less
than the number of states of $\B$. The computation of the quotient
automaton (performed by the Procedure \textsc{Merge}) is described in
Section~\ref{section.pseudocode}.

\section{Finding a stable pair} \label{section.algo2}

In this section, we consider an aperiodic complete deterministic
automaton $\A$ over the alphabet $A$. We design a linear-time
algorithm for finding an equivalent automaton which has a stable pair.

In order to describe the algorithm, we give some definitions and
notation.

Let $\R$ be the subgraph of the graph of $\A$ made of the red edges.
The graph $\R$ is a disjoint union of connected components called
\emph{clusters}. Since each state has exactly one outgoing edge in
$\R$, each cluster contains a unique (red) cycle with trees attached to the
cycle at their roots. If $r$ is the root of such a tree, its
\emph{children} are the states $p$ such that $p$ is not on the a red
cycle and $(p,a,r)$ is an edge. If $p,q$ belong to the same tree, $p$ is
an \emph{ancestor} of $q$ (or $q$ is a \emph{descendant} of $p$) in
the tree if there is a red path from $q$ to $p$. Note that in these trees,
the edges are oriented from the child to the parents and the paths from
the descendant to the ancestors.

If $q$ belongs to some red cycle of length greater than $1$, its
\emph{predecessor} is the unique state $p$ belonging to the same cycle
such that $(p,a,q)$ is an edge. In the case the length of the cycle
is $1$, we set that the predecessor is $q$ itself.

For each state $p$ belonging to some cluster, we define the \emph{level} of $p$
as the distance between $p$ and the root of the tree containing
$p$. If $p$ belongs to the cycle of the cluster, its level is thus null.
The \emph{level of an automaton} is the maximal level of its states.
A \emph{maximal state} is a state of maximal level.  
A \emph{maximal tree} is a tree containing at least one maximal state and rooted at
a state of level $0$. 
 A \emph{maximal root} is the root of a maximal
tree and a \emph{maximal child} of a maximal root $r$ is a child of $r$ having at least one maximal state as descendant.



The algorithm for finding a coloring which has a stable pair relies on
the following key lemma due to Trahtman~\cite{Trahtman09}. It uses the
notion of minimal images in an automaton.  An \emph{image} in an
automaton $\A=(Q,E)$ is a set of states $I = Q \cdot w$, where $w$ is
a word and $Q \cdot w = \{q \cdot w \mid q \in Q\}$.  A \emph{minimal
  image} in an automaton is an image which does not properly contain
another image. In an irreducible automaton two minimal images have
the same cardinality which is called the \emph{minimal rank} of~$\A$.
Also, if $I$ is a minimal image and $u$ is a word, then $I \cdot u$ is
again a minimal image and the map $p \rightarrow p \cdot u$ is
one-to-one from $I$ onto $I \cdot u$.

Note that the hypotheses in the statement below depend on the choice
of the letter $a$ defining the red edges.
\begin{lemma}[Trahtman~\cite{Trahtman09}] \label{lemma.sameTree} Let
  $\A$ be an irreducible complete deterministic automaton with a
  positive level.  If all maximal states in $\A$ belong to the same
  tree, then $\A$ has a stable pair.
\end{lemma}

\begin{proof}
  Since $\A$ is irreducible, there is a minimal image $I$ containing a
  maximal state $p$. Let $\ell > 0$ the level of $p$ (\ie the distance
  between $p$ and the root $r$ of the unique maximal tree). Let us assume
  that there is a state $q \neq p$ in $I$ of level $\ell$.  Then the
  cardinal of $I \cdot a^{\ell}$ is strictly less than the cardinal of
  $I$, which contradicts the minimality of $I$. Thus all states but
  $p$ in $I$ have level strictly less than~$\ell$.

  Let $m$ be a common multiple of the lengths of all red cycles. Let
  $C$ be the red cycle containing $r$. Let $s_0$ be the predecessor
  of $r$ in $C$ and $s_1$ the child of $r$ containing $p$ in its
  subtree. Since $\ell > 0$, we have $s_0 \neq s_1$.  Let $J = I \cdot
  a^{\ell-1}$ and $K = J \cdot a^m$.  Since the level of all states of
  $I$ but $p$ is less than or equal to $\ell-1$, the set $J$ is equal
  to $\{s_1\} \cup R$, where $R$ is a set of states belonging to the
  red cycles.  Since for any state $q$ in a red cycle, $q \cdot a^m
  =q$, we get $K = \{s_0\} \cup R$.
 
  Let $w$ be a word such that $Q \cdot w$ is a minimal image. For any
  word $v$, the minimal images $J \cdot vw$ and $K \cdot vw$ have the
  same cardinal equal to the cardinal of $I$.  We claim that the set
  $(J \cup K) \cdot vw$ is a minimal image. Indeed, $J \cdot vw
  \subseteq (J \cup K) \cdot vw \subseteq Q \cdot vw$, hence all three
  are equal. But $(J \cup K) \cdot vw = R \cdot vw \cup s_0 \cdot vw
  \cup s_1 \cdot vw$.  This forces $s_0 \cdot vw = s_1 \cdot vw$ since
  the cardinality of $R \cdot vw$ cannot be less than the cardinality
  of $R$. As a consequence $(s_0 \cdot v,s_1 \cdot v)$ is
  synchronizable and thus $(s_0,s_1)$ is a stable pair.
\end{proof}

In the sequel, we call Condition~$\C$ the assumption of
Lemma \ref{lemma.sameTree}: \emph{all maximal states belong to the same
  tree}.  

In the subsections below, we describe sequences of flips of edges
that make the resulting equivalent automaton satisfy Condition $\C$
and hence have a stable pair. We consider several cases corresponding
to the geometry of the automaton.

\subsection{The case of null maximal level}

In this section, we assume that the level of the automaton is $\ell
=0$. The subgraph $\R$ of red edges is a disjoint union of cycles.  

A set of edges going out of a state $p$ is called a \emph{bunch} if
these edges all end in a same state $q$. Note that if a state $q$ has
two incoming bunches from two states $p,p'$, then $(p,p')$ is a
stable pair.

If the set of outgoing edges of each state is a bunch, then there is
only one red cycle, and the automaton is not aperiodic unless the
trivial case where the length of this cycle is $1$. We can thus assume
that there is a state $p$ whose set of outgoing edges is not a
bunch. There exists $b \neq a$ and $q \neq r$ such that $(p,a,q)$ and
$(p,b,r)$ are edges. We flip these two edges. This gives an automaton $\A$
which satisfies Condition~$\C$.  Let $s$ be the state which is the
predecessor of $r$ in its red cycle.  It follows from the proof of
Lemma~\ref{lemma.sameTree} that the pair $(p,s)$ is a stable pair.

This case is described in the pseudocode \textsc{LevelZeroFlipEdges}
where \textsc{GetPredecessor}$(r)$ returns the predecessor of $r$ on its red cycle.  The function \textsc{LevelZero\-FlipEdges}$(\A)$
returns an automaton equivalent to $\A$ together with a stable pair.

\begin{small} 
\begin{codebox}
\Procname{$\proc{LevelZeroFlipEdges }(\text{automaton } \A \text{ of level } \ell = 0$)}
\li    \For each state  $p$ on a red cycle $C$     
\li        \Do \If the set of outgoing edges of $p$ is not a bunch
\li                \Then  let $e = (p,a,q)$ and $f = (p,b,r)$ be edges with $b \neq a$ and $q \neq r$
\li                       \proc{Flip}$(e,f)$
\li                       $s \gets$ \proc{GetPredecessor}$(r)$
\li                      \Return $\A$, $(p,s)$
               \End
       \End 
\li    \Return \proc{Error}($\A$ is not aperiodic)  
\end{codebox}
\end{small}

The procedure \textsc{Flip}$(e,f)$
exchanges the labels of two edges $e,f$. It also performs the corresponding update of data as explained in Section~\ref{section.update}. 

\subsection{The case of non-null maximal level}

In this section, we assume that the level of the automaton is $\ell > 0$.

\subsubsection{Main treatment}




We describe a sequence of flips of edges such that the automaton
obtained after this sequence of flips has a unique maximal tree. Note
that the levels and other useful data will not be recomputed after
each flip (which would increase the time complexity too much).

Let $C$ be a red cycle containing a maximal tree $\T$ rooted at $r$.
We denote by $r_1=r,r_2, \dotsc r_k$ the maximal roots of $C$ in the
order given by the orientation of the red edges of the cycle.  For $k
> 1$ and $1 \leq i \leq k$ we denote by $I(r_{i})$ the set of states
contained in the red simple path from the root $r_j$ with $j= ({i-1}
\bmod k) +1$ to $r_i$ with $r_j$ included and $r_{i}$ excluded.  For
$k = 1$ we define $I(r)$ as the set of all states of $C$.
Similarly, for  $k > 1$ and $1 \leq i \leq k$ we denote by $J(r_{i})$
the set of states contained in the red simple path from the root $r_j$
with $j= ({i-1} \bmod k) +1$ to $r_i$ with $r_j$ excluded and $r_{i}$
included.  For $k = 1$ we define $J(r)$ as the set of all states of
$C$.

We denote by $s_0$ the predecessor of $r$ in $C$. If the length of
$C$ is $1$, $s_0=r$. We denote by $S(r)$ the set of maximal children
of $r$ (\ie which are ancestors of some maximal state). Let $\rho$ be
the cardinality of $S(r)$. For each $s$ in $S(r)$, we choose a maximal
state $p$ in the subtree rooted at $s$ (see Fig.
\ref{figure.tree}). There may be several possible choices for the
state $p$ and we select one of them arbitrarily. We denote by $P(r)$
the set of these maximal states. This set has cardinality $\rho$.

The key idea, in order to guarantee the global linear complexity, is
to perform operations for each maximal root $r$, whose time complexity
is linear in the number of nodes belonging to trees attached to the
states contained in $J(r)$.

\begin{figure}[htbp]
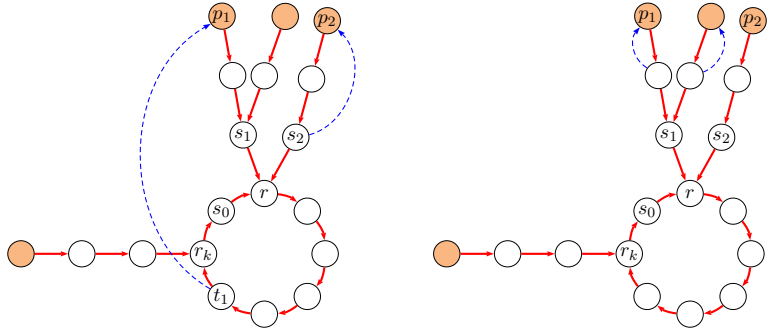

    \centering
\VCDraw{%
\begin{VCPicture}{(0,3)(12,14)}
\MediumState
\VCPut{(2,6)}{
\State[]{(2,0)}{1}%
\ChgStateFillColor{Apricot}%
 \RstStateFillColor
\VCPut[45]{(0,0)}{%
    \State[]{(2,0)}{2}
}
\VCPut[90]{(0,0)}{%
    \State[r]{(2,0)}{3}
}
\VCPut[75]{(0,0)}{%
    \State[s_2]{(4,0)}{11}%
    \State[]{(6,0)}{12}%
    \ChgStateFillColor{Apricot}%
    \State[p_2]{(8,0)}{18}%
    \RstStateFillColor
}
\VCPut[90]{(0,0)}{%
    \State[]{(5.9,0)}{19}
}
\VCPut[85]{(0,0)}{%
    \ChgStateFillColor{Apricot}%
    \State[]{(7.9,0)}{20}%
    \RstStateFillColor
}
\VCPut[100]{(0,0)}{%
    \State[s_1]{(4,0)}{13}%
    \State[]{(6,0)}{14}%
    \ChgStateFillColor{Apricot}
    \State[p_1]{(8,0)}{16}%
    \RstStateFillColor  
}
\VCPut[135]{(0,0)}{%
    \State[s_0]{(2,0)}{4}}
\VCPut[180]{(0,0)}{%
    \State[r_k]{(2,0)}{5}
    \State[]{(4,0)}{21}
    \State[]{(6,0)}{22}
    \ChgStateFillColor{Apricot} \State[]{(8,0)}{23}  \RstStateFillColor  
}
\VCPut[225]{(0,0)}{%
    \State[t_1]{(2,0)}{6} 
}
\VCPut[270]{(0,0)}{%
    \State[]{(2,0)}{7}
}
\VCPut[315]{(0,0)}{%
    \State[]{(2,0)}{8}}
\ChgEdgeLineWidth{2}
\ChgEdgeLineColor{red}
\ChgEdgeLineStyle{solid}
\ArcL[.5]{1}{8}{}
\ArcL[.5]{2}{1}{}
\ArcL[.5]{3}{2}{}
\ArcL[.5]{4}{3}{}
\ArcL[.5]{5}{4}{}
\ArcL[.5]{6}{5}{}
\ArcL[.5]{7}{6}{}
\ArcL[.5]{8}{7}{}
\EdgeL[.5]{10}{9}{}
\EdgeL[.5]{9}{1}{}
\EdgeL[.5]{12}{11}{}
\EdgeL[.5]{11}{3}{}
\EdgeL[.5]{14}{13}{}
\EdgeL[.5]{13}{3}{}
\EdgeL[.5]{16}{14}{}
\EdgeL[.5]{18}{12}{}
\EdgeL[.5]{19}{13}{}
\EdgeL[.5]{20}{19}{}
\EdgeL[.5]{23}{22}{}
\EdgeL[.5]{22}{21}{}
\EdgeL[.5]{21}{5}{}
\ChgEdgeLineColor{blue}
\ChgEdgeLineWidth{1}
\ChgEdgeLineStyle{dashed}
\VArcR[.5]{arcangle=60}{6}{16}{}
\VArcR[.5]{arcangle=-60}{11}{18}{}
}
\VCPut{(16,6)}{
\State[]{(2,0)}{1}%
\ChgStateFillColor{Apricot}%
 \RstStateFillColor
\VCPut[45]{(0,0)}{%
    \State[]{(2,0)}{2}
}
\VCPut[90]{(0,0)}{%
    \State[r]{(2,0)}{3}
}
\VCPut[75]{(0,0)}{%
    \State[s_2]{(4,0)}{11}%
    \State[]{(6,0)}{12}%
    \ChgStateFillColor{Apricot}%
    \State[p_2]{(8,0)}{18}%
    \RstStateFillColor
}
\VCPut[90]{(0,0)}{%
    \State[]{(5.9,0)}{19}
}
\VCPut[85]{(0,0)}{%
    \ChgStateFillColor{Apricot}%
    \State[]{(7.9,0)}{20}%
    \RstStateFillColor
}
\VCPut[100]{(0,0)}{%
    \State[s_1]{(4,0)}{13}%
    \State[]{(6,0)}{14}%
    \ChgStateFillColor{Apricot}
    \State[p_1]{(8,0)}{16}%
    \RstStateFillColor  
}
\VCPut[135]{(0,0)}{%
    \State[s_0]{(2,0)}{4}}
\VCPut[180]{(0,0)}{%
    \State[r_k]{(2,0)}{5}
    \State[]{(4,0)}{21}
    \State[]{(6,0)}{22}
    \ChgStateFillColor{Apricot} \State[]{(8,0)}{23}  \RstStateFillColor  
}
\VCPut[225]{(0,0)}{%
    \State[]{(2,0)}{6} 
}
\VCPut[270]{(0,0)}{%
    \State[]{(2,0)}{7}
}
\VCPut[315]{(0,0)}{%
    \State[]{(2,0)}{8}}
\ChgEdgeLineWidth{2}
\ChgEdgeLineColor{red}
\ChgEdgeLineStyle{solid}
\ArcL[.5]{1}{8}{}
\ArcL[.5]{2}{1}{}
\ArcL[.5]{3}{2}{}
\ArcL[.5]{4}{3}{}
\ArcL[.5]{5}{4}{}
\ArcL[.5]{6}{5}{}
\ArcL[.5]{7}{6}{}
\ArcL[.5]{8}{7}{}
\EdgeL[.5]{10}{9}{}
\EdgeL[.5]{9}{1}{}
\EdgeL[.5]{12}{11}{}
\EdgeL[.5]{11}{3}{}
\EdgeL[.5]{14}{13}{}
\EdgeL[.5]{13}{3}{}
\EdgeL[.5]{16}{14}{}
\EdgeL[.5]{18}{12}{}
\EdgeL[.5]{19}{13}{}
\EdgeL[.5]{20}{19}{}
\EdgeL[.5]{23}{22}{}
\EdgeL[.5]{22}{21}{}
\EdgeL[.5]{21}{5}{}
\ChgEdgeLineColor{blue}
\ChgEdgeLineWidth{1}
\ChgEdgeLineStyle{dashed}
\VArcR[.5]{arcangle=50}{14}{16}{}
\VArcR[.5]{arcangle=-50}{19}{20}{}
}
\end{VCPicture}%
        }
        \caption{On the left part the figure, the dashed (blue) edge
          ending in $p_1$ has type 1 while the one ending in $p_2$ has
          type 3. On the right part, the set $L_{s_1}$ of dashed
          (blue) edges cover all maximal states of the subtree rooted
          at the child $s_1$.  }\label{figure.tree}
\end{figure}

Since the automaton is irreducible, for each $p \in P(r)$ there is at
least one blue edge ending in $p$. Each blue edge $(t,b,p)$ ending in
a state $p \in P(r)$ can be of one of the following type depending
on the position of $t$ in the graph:
\begin{itemize}
\item type 0: $t$ is not in the same cluster as $r$, or $t$ has a
  positive level and $t$ is not an ancestor of $p$ in $\T$.
\item type 1: $t$ is in the same cluster as $r$, has a null level, and
  $t$ is outside the interval $I(r)$.
\item type 2: $t$ is in the same cluster as $r$, has a null level, and
  $t$ is contained in the interval $I(r)$. This includes the
  particular case where $k=1$ and $t=r$.
\item type 3: $t$ is an ancestor of $p$ in $\T$ and $t \neq r$.
\end{itemize}
Note that it is possible that $t=p$. In this case the edge $(t,b,p)$
has type 0 since $t$ has a positive level.

A procedure \textsc{FindEdges}$(r)$, that will be described later in
detail (see Section~\ref{section.auxiliary}), first flips some edges and returns
a value of one of the following forms.
\begin{itemize} 
\item A pair $(0,e)$, where $e$ is an edge of type 0.
\item A triple $(1,e,f)$, where $e,f$ are two edges of type 1 or 2
  ending in distinct states of $P(r)$.
\item A pair $(2,e)$, where $e$ is an edge of type 1 or 2.
Moreover, in this case, the procedure modifies the tree $\T$ in such a way
that $r$ has a unique maximal child.
\item A pair $(3,e)$, where $e$ is an edge of type 3 starting at a
  state which is an ancestor of all maximal nodes of $\T$.
\end{itemize}   

For each maximal root $r$, the procedure \textsc{FlipEdges}$(\A,r)$
returns either an automaton equivalent to $\A$ together with a stable
pair, or an automaton equivalent to $\A$ together with one edge
$(t_r,b_r,p_r)$.  Its execution depends on the value returned by
\textsc{FindEdges}$(r)$ according to the following four cases
described below. After running \textsc{FlipEdges}$(\A,r)$ on each
maximal root, we obtain either an automaton satisfying Condition~$\C$
(\ie which has a stable pair) or an automaton where each maximal root
$r$ has a unique maximal child and such that the potential flip of
$(t_r,b_r,p_r)$ with the red edge starting at $t_r$ makes the root $r$
not maximal anymore. In the first case, our goal is achieved. In the
latter case, we flip the blue edge $(t_r,b_r,p_r)$ and the red one
starting at $t_r$ for all maximal roots $r$ \emph{but one}. We get an
equivalent automaton which has unique maximal tree and thus has a
stable pair by Lemma~\ref{lemma.sameTree}. The combination of all
these transformations is realized by the procedure
\textsc{FindStablePair} given at the end of this section.

The possible values returned by the procedure \textsc{FlipEdges}$(\A,r)$
are the following.
\begin{itemize}
\item Case 0. The value returned by \textsc{FindEdges}$(r)$ is $(0,e)$
  with $e=(t_1,b_1,p_1)$ of type 0. The procedure
  \textsc{FlipEdges}$(\A,r)$ returns the automaton obtained by
  flipping the edge $(t_1,b_1,p_1)$ and the red edge going out of
  $t_1$. This automaton is equivalent to $\A$ and satisfies
  Condition~$\C$. Indeed, one may easily check that, after the flip,
  all states of maximal level belong to the same tree as $p_1$.
%

\begin{figure}[htbp]
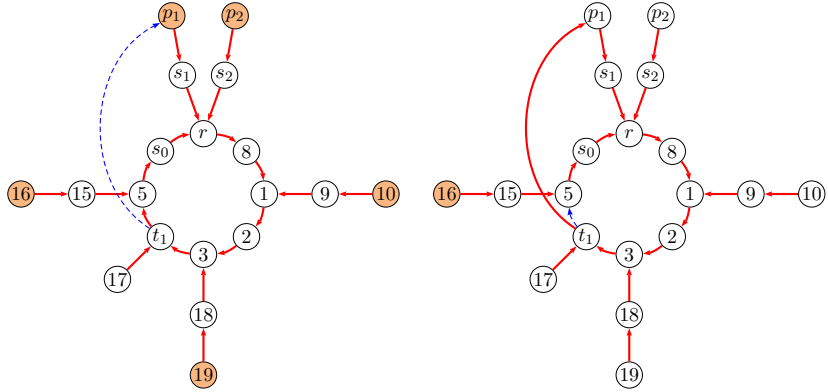

    \centering
\VCDraw{%
\begin{VCPicture}{(0,-1)(12,13)}
\MediumState
\VCPut{(0,6)}{
\State[1]{(2,0)}{1}%
\State[9]{(4,0)}{9}%
\ChgStateFillColor{Apricot}%
\State[10]{(6,0)}{10}%
 \RstStateFillColor
\VCPut[45]{(0,0)}{%
    \State[8]{(2,0)}{2}
}
\VCPut[90]{(0,0)}{%
    \State[r]{(2,0)}{3}
}
\VCPut[80]{(0,0)}{%
    \State[s_2]{(4,0)}{11}%
    \ChgStateFillColor{Apricot}%
    \State[p_2]{(6,0)}{12}%
    \RstStateFillColor
}
\VCPut[100]{(0,0)}{%
    \State[s_1]{(4,0)}{13}%
    \ChgStateFillColor{Apricot}
    \State[p_1]{(6,0)}{14}%
    \RstStateFillColor  
}
\VCPut[135]{(0,0)}{%
    \State[s_0]{(2,0)}{4}}
\VCPut[180]{(0,0)}{%
    \State[5]{(2,0)}{5}
    \State[15]{(4,0)}{15}%
    \ChgStateFillColor{Apricot}
    \State[16]{(6,0)}{16}%
    \RstStateFillColor
}
\VCPut[225]{(0,0)}{%
    \State[t_1]{(2,0)}{6}
    \State[17]{(4,0)}{17}
}
\VCPut[270]{(0,0)}{%
    \State[3]{(2,0)}{7}
    \State[18]{(4,0)}{18}%
    \ChgStateFillColor{Apricot}
    \State[19]{(6,0)}{19}%
    \RstStateFillColor 
}
\VCPut[315]{(0,0)}{%
    \State[2]{(2,0)}{8}}
\ChgEdgeLineWidth{2}
\ChgEdgeLineColor{red}
\ChgEdgeLineStyle{solid}
\ArcL[.5]{1}{8}{}
\ArcL[.5]{2}{1}{}
\ArcL[.5]{3}{2}{}
\ArcL[.5]{4}{3}{}
\ArcL[.5]{5}{4}{}
\ArcL[.5]{6}{5}{}
\ArcL[.5]{7}{6}{}
\ArcL[.5]{8}{7}{}
\EdgeL[.5]{10}{9}{}
\EdgeL[.5]{9}{1}{}
\EdgeL[.5]{12}{11}{}
\EdgeL[.5]{11}{3}{}
\EdgeL[.5]{14}{13}{}
\EdgeL[.5]{13}{3}{}
\EdgeL[.5]{16}{15}{}
\EdgeL[.5]{15}{5}{}
\EdgeL[.5]{17}{6}{}
\EdgeL[.5]{19}{18}{}
\EdgeL[.5]{18}{7}{}
\ChgEdgeLineColor{blue}
\ChgEdgeLineWidth{1}
\ChgEdgeLineStyle{dashed}
\VArcR[.5]{arcangle=60}{6}{14}{}
}
\VCPut{(14,6)}{
\State[1]{(2,0)}{1}%
\State[9]{(4,0)}{9}%
\State[10]{(6,0)}{10}%
\VCPut[45]{(0,0)}{%
    \State[8]{(2,0)}{2}
}
\VCPut[90]{(0,0)}{%
    \State[r]{(2,0)}{3}
}
\VCPut[80]{(0,0)}{%
    \State[s_2]{(4,0)}{11}%
    \State[p_2]{(6,0)}{12}%
}
\VCPut[100]{(0,0)}{%
    \State[s_1]{(4,0)}{13}%
    \State[p_1]{(6,0)}{14}%
}
\VCPut[135]{(0,0)}{%
    \State[s_0]{(2,0)}{4}}
\VCPut[180]{(0,0)}{%
    \State[5]{(2,0)}{5}
    \State[15]{(4,0)}{15}%
    \ChgStateFillColor{Apricot}
    \State[16]{(6,0)}{16}%
    \RstStateFillColor
}
\VCPut[225]{(0,0)}{%
    \State[t_1]{(2,0)}{6}
    \State[17]{(4,0)}{17}
}
\VCPut[270]{(0,0)}{%
    \State[3]{(2,0)}{7}
    \State[18]{(4,0)}{18}%
    \State[19]{(6,0)}{19}%
}
\VCPut[315]{(0,0)}{%
    \State[2]{(2,0)}{8}}
\ChgEdgeLineWidth{2}%
\ChgEdgeLineColor{red}%
\ChgEdgeLineStyle{solid}%
\ArcL[.5]{1}{8}{}
\ArcL[.5]{2}{1}{}
\ArcL[.5]{3}{2}{}
\ArcL[.5]{4}{3}{}
\ArcL[.5]{5}{4}{}
\ArcL[.5]{7}{6}{}
\ArcL[.5]{8}{7}{}
\EdgeL[.5]{10}{9}{}
\EdgeL[.5]{9}{1}{}
\EdgeL[.5]{12}{11}{}
\EdgeL[.5]{11}{3}{}
\EdgeL[.5]{14}{13}{}
\EdgeL[.5]{13}{3}{}
\EdgeL[.5]{16}{15}{}
\EdgeL[.5]{15}{5}{}
\EdgeL[.5]{17}{6}{}
\EdgeL[.5]{19}{18}{}
\EdgeL[.5]{18}{7}{}
\VArcR[.5]{arcangle=60}{6}{14}{}
\ChgEdgeLineColor{blue}%
\ChgEdgeLineWidth{1}%
\ChgEdgeLineStyle{dashed}%
\ArcL[.5]{6}{5}{}%
}
\end{VCPicture}%
        }
        \caption{The picture on the left illustrates Case 1.1.
          The edge $(t_1,b_1,p_1)$ if of type 1. After flipping the
          edge $(t_1,b_1,p_1)$ and the red edge going out of $t_1$, we
          get the automaton on the right. It satisfies the
          Condition~$\C$, \ie it has a unique maximal tree (here
          rooted at $r$). Maximal states are colored and the (dashed)
          $b$-edges of the automaton are not all
          represented.}\label{figure.cluster1}
\end{figure}
\item Case 1. The value returned by \textsc{FindEdges}$(r)$ is
$(1,e_1,e_2)$, with $e_1=(t_1,b_1,p_1)$, $e_2=(t_2,b_2,p_2)$ of
type 1 or 2. Recall that $p_1 \neq p_2$ and that $b_1,b_2 \neq a$.
\begin{itemize}
\item Case 1.1. If $e_1$ (or $e_2$) has type 1, the same
conclusion as in Case 0 holds by flipping the edge $(t_1,b_1,p_1)$ and the red edge
going out of~$t_1$, as is shown in Fig.~\ref{figure.cluster1}.
\item Case 1.2. In the case both edges $e_1,e_2$ have type 2 and $t_1
  \neq t_2$, without loss of generality, we may assume that $t_1 <
  t_2$ in the interval $I(r)$ (see Fig.~\ref{figure.cluster2}).  We
      flip the edge $(t_1,b_1,p_1)$ and the red edge going out of
      $t_1$.  We denote by $\T'$ the tree rooted at~$r$ after this
      flip.
\begin{itemize}
\item Case 1.2.1.
If the height of $\T'$
    is greater than $\ell$, the automaton satisfies Condition~$\C$ (see
    the right part of~Fig.~\ref{figure.cluster2}). 
\item Case 1.2.2.  Otherwise the height of $\T'$ is at most $\ell$
  (see the left part of~Fig.~\ref{figure.cluster3}).  In that case,
  we also flip the edge $(t_2,b_2,p_2)$ and the red edge going out of
  $t_2$. The new equivalent automaton satisfies Condition~$\C$ (see
  the right part of~Fig.~\ref{figure.cluster3}).  The computation of
  the size of $\T'$ is detailed in Section~\ref{section.complexity}.
\end{itemize}

\begin{figure}[htbp]
    \centering
\VCDraw{%
\begin{VCPicture}{(0,-1)(12,12)}
\MediumState
\VCPut{(0,6)}{
\State[1]{(2,0)}{1}%
\State[9]{(4,0)}{9}%
\ChgStateFillColor{Apricot}
\State[10]{(6,0)}{10}%
 \RstStateFillColor
\VCPut[45]{(0,0)}{%
    \State[8]{(2,0)}{2}
}
\VCPut[90]{(0,0)}{%
    \State[r]{(2,0)}{3}
}
\VCPut[60]{(0,0)}{%
    \State[21]{(4,0)}{21}%
}
\VCPut[80]{(0,0)}{%
    \State[s_2]{(4,0)}{11}%
    \ChgStateFillColor{Apricot}
    \State[p_2]{(6,0)}{12}%
    \RstStateFillColor
}
\VCPut[100]{(0,0)}{%
    \State[s_1]{(4,0)}{13}%
    \ChgStateFillColor{Apricot}
    \State[p_1]{(6,0)}{14}%
    \RstStateFillColor  
}
\VCPut[135]{(0,0)}{%
    \State[t_2]{(2,0)}{4}
    \State[20]{(4,0)}{20}%
}
\VCPut[180]{(0,0)}{%
    \State[5]{(2,0)}{5}
    \State[15]{(4,0)}{15}%
}
\VCPut[225]{(0,0)}{%
    \State[t_1]{(2,0)}{6}
    \State[17]{(4,0)}{17}
    \ChgStateFillColor{Apricot}   
     \State[22]{(6,0)}{22}
    \RstStateFillColor
}
\VCPut[270]{(0,0)}{%
    \State[3]{(2,0)}{7}
    \State[18]{(4,0)}{18}%
    \ChgStateFillColor{Apricot}
    \State[19]{(6,0)}{19}%
    \RstStateFillColor 
}
\VCPut[315]{(0,0)}{%
    \State[2]{(2,0)}{8}}
\ChgEdgeLineWidth{2}
\ChgEdgeLineColor{red}
\ChgEdgeLineStyle{solid}
\ArcL[.5]{1}{8}{}
\ArcL[.5]{2}{1}{}
\ArcL[.5]{3}{2}{}
\ArcL[.5]{4}{3}{}
\ArcL[.5]{5}{4}{}
\ArcL[.5]{6}{5}{}
\ArcL[.5]{7}{6}{}
\ArcL[.5]{8}{7}{}
\EdgeL[.5]{10}{9}{}
\EdgeL[.5]{9}{1}{}
\EdgeL[.5]{12}{11}{}
\EdgeL[.5]{11}{3}{}
\EdgeL[.5]{14}{13}{}
\EdgeL[.5]{13}{3}{}
\EdgeL[.5]{15}{5}{}
\EdgeL[.5]{17}{6}{}
\EdgeL[.5]{19}{18}{}
\EdgeL[.5]{18}{7}{}
\EdgeL[.5]{20}{4}{}
\EdgeL[.5]{22}{17}{}
\EdgeL[.5]{21}{3}{}
\ChgEdgeLineColor{blue}
\ChgEdgeLineWidth{1}
\ChgEdgeLineStyle{dashed}
\VArcR[.5]{arcangle=70,ncurv=0.9}{6}{14}{}
\VArcR[.5]{arcangle=50}{4}{12}{}
}
\VCPut{(14,6)}{

\State[1]{(2,0)}{1}%
\State[9]{(4,0)}{9}%
\State[10]{(6,0)}{10}%
\VCPut[45]{(0,0)}{%
    \State[8]{(2,0)}{2}
}
\VCPut[90]{(0,0)}{%
    \State[r]{(2,0)}{3}
}
\VCPut[60]{(0,0)}{%
    \State[21]{(4,0)}{21}%
}
\VCPut[80]{(0,0)}{%
    \State[s_2]{(4,0)}{11}%
    \State[p_2]{(6,0)}{12}%
}
\VCPut[100]{(0,0)}{%
    \State[s_1]{(4,0)}{13}%
    \State[p_1]{(6,0)}{14}%
}
\VCPut[135]{(0,0)}{%
    \State[t_2]{(2,0)}{4}
    \State[20]{(4,0)}{20}%
}
\VCPut[180]{(0,0)}{%
    \State[5]{(2,0)}{5}   
    \ChgStateFillColor{Apricot}   
    \State[15]{(4,0)}{15}%
   \RstStateFillColor
}
\VCPut[225]{(0,0)}{%
    \State[t_1]{(2,0)}{6}
    \State[17]{(4,0)}{17}
     \State[22]{(6,0)}{22}
}
\VCPut[270]{(0,0)}{%
    \State[3]{(2,0)}{7}
    \State[18]{(4,0)}{18}%
    \State[19]{(6,0)}{19}%
}
\VCPut[315]{(0,0)}{%
    \State[2]{(2,0)}{8}}
\ChgEdgeLineWidth{2}
\ChgEdgeLineColor{red}
\ChgEdgeLineStyle{solid}
\ArcL[.5]{1}{8}{}
\ArcL[.5]{2}{1}{}
\ArcL[.5]{3}{2}{}
\ArcL[.5]{4}{3}{}
\ArcL[.5]{7}{6}{}
\ArcL[.5]{8}{7}{}
\EdgeL[.5]{10}{9}{}
\EdgeL[.5]{9}{1}{}
\EdgeL[.5]{12}{11}{}
\EdgeL[.5]{11}{3}{}
\EdgeL[.5]{14}{13}{}
\EdgeL[.5]{13}{3}{}
\EdgeL[.5]{15}{5}{}
\EdgeL[.5]{17}{6}{}
\EdgeL[.5]{19}{18}{}
\EdgeL[.5]{18}{7}{}
\EdgeL[.5]{20}{4}{}
\EdgeL[.5]{22}{17}{}
\EdgeL[.5]{21}{3}{}
\ArcL[.5]{5}{4}{}
\VArcR[.5]{arcangle=70,ncurv=0.9}{6}{14}{}
\ChgEdgeLineColor{blue}
\ChgEdgeLineWidth{1}
\ChgEdgeLineStyle{dashed}
\VArcR[.5]{arcangle=50}{4}{12}{}
\ArcL[.5]{6}{5}{}
}
\end{VCPicture}%
        }
        \caption{The picture on the left illustrates Case 1.2.1 of
          the main treatment. There are two edges $(t_1,b_1,p_1)$,
          $(t_2,b_2,p_2)$ of type 2. The height of the tree $\T'$
          obtained after flipping the edge $(t_1,b_1,p_1)$ and the red
          edge going out of $t_1$, is 3, which is greater than the maximal
          level. We get a unique maximal tree rooted at $r$ in the
          same cluster.  The picture on the right illustrates the
          result.}
\label{figure.cluster2}
\end{figure}
\begin{figure}[h]
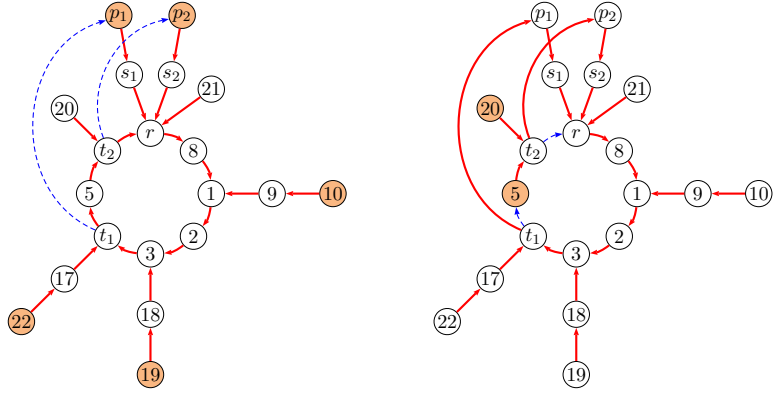

    \centering
\VCDraw{%
\begin{VCPicture}{(0,-1)(12,12)}
\MediumState
\VCPut{(0,6)}{
\State[1]{(2,0)}{1}%
\State[9]{(4,0)}{9}%
\ChgStateFillColor{Apricot}
\State[10]{(6,0)}{10}%
 \RstStateFillColor
\VCPut[45]{(0,0)}{%
    \State[8]{(2,0)}{2}
}
\VCPut[90]{(0,0)}{%
    \State[r]{(2,0)}{3}
}
\VCPut[60]{(0,0)}{%
    \State[21]{(4,0)}{21}%
}
\VCPut[80]{(0,0)}{%
    \State[s_2]{(4,0)}{11}%
    \ChgStateFillColor{Apricot}
    \State[p_2]{(6,0)}{12}%
    \RstStateFillColor
}
\VCPut[100]{(0,0)}{%
    \State[s_1]{(4,0)}{13}%
    \ChgStateFillColor{Apricot}
    \State[p_1]{(6,0)}{14}%
    \RstStateFillColor  
}
\VCPut[135]{(0,0)}{%
    \State[t_2]{(2,0)}{4}
    \State[20]{(4,0)}{20}%
}
\VCPut[180]{(0,0)}{%
    \State[5]{(2,0)}{5}
}
\VCPut[225]{(0,0)}{%
    \State[t_1]{(2,0)}{6}
    \State[17]{(4,0)}{17}
    \ChgStateFillColor{Apricot}   
     \State[22]{(6,0)}{22}
    \RstStateFillColor
}
\VCPut[270]{(0,0)}{%
    \State[3]{(2,0)}{7}
    \State[18]{(4,0)}{18}%
    \ChgStateFillColor{Apricot}
    \State[19]{(6,0)}{19}%
    \RstStateFillColor 
}
\VCPut[315]{(0,0)}{%
    \State[2]{(2,0)}{8}}
\ChgEdgeLineWidth{2}
\ChgEdgeLineColor{red}
\ChgEdgeLineStyle{solid}
\ArcL[.5]{1}{8}{}
\ArcL[.5]{2}{1}{}
\ArcL[.5]{3}{2}{}
\ArcL[.5]{4}{3}{}
\ArcL[.5]{5}{4}{}
\ArcL[.5]{6}{5}{}
\ArcL[.5]{7}{6}{}
\ArcL[.5]{8}{7}{}
\EdgeL[.5]{10}{9}{}
\EdgeL[.5]{9}{1}{}
\EdgeL[.5]{12}{11}{}
\EdgeL[.5]{11}{3}{}
\EdgeL[.5]{14}{13}{}
\EdgeL[.5]{13}{3}{}
\EdgeL[.5]{17}{6}{}
\EdgeL[.5]{19}{18}{}
\EdgeL[.5]{18}{7}{}
\EdgeL[.5]{20}{4}{}
\EdgeL[.5]{22}{17}{}
\EdgeL[.5]{21}{3}{}
\ChgEdgeLineColor{blue}
\ChgEdgeLineWidth{1}
\ChgEdgeLineStyle{dashed}
\VArcR[.5]{arcangle=70,ncurv=0.9}{6}{14}{}
\VArcR[.5]{arcangle=50}{4}{12}{}
}
\VCPut{(14,6)}{

\State[1]{(2,0)}{1}%
\State[9]{(4,0)}{9}%
\State[10]{(6,0)}{10}%
\VCPut[45]{(0,0)}{%
    \State[8]{(2,0)}{2}
}
\VCPut[90]{(0,0)}{%
    \State[r]{(2,0)}{3}
}
\VCPut[60]{(0,0)}{%
    \State[21]{(4,0)}{21}%
}
\VCPut[80]{(0,0)}{%
    \State[s_2]{(4,0)}{11}%
    \State[p_2]{(6,0)}{12}%
}
\VCPut[100]{(0,0)}{%
    \State[s_1]{(4,0)}{13}%
    \State[p_1]{(6,0)}{14}%
}
\VCPut[135]{(0,0)}{%
    \State[t_2]{(2,0)}{4}
    \ChgStateFillColor{Apricot}
    \State[20]{(4,0)}{20}%
    \RstStateFillColor 
}
\VCPut[180]{(0,0)}{%
    \ChgStateFillColor{Apricot}  
     \State[5]{(2,0)}{5}  
    \RstStateFillColor 
}
\VCPut[225]{(0,0)}{%
    \State[t_1]{(2,0)}{6}
    \State[17]{(4,0)}{17}
     \State[22]{(6,0)}{22}
}
\VCPut[270]{(0,0)}{%
    \State[3]{(2,0)}{7}
    \State[18]{(4,0)}{18}%
    \State[19]{(6,0)}{19}%
}
\VCPut[315]{(0,0)}{%
    \State[2]{(2,0)}{8}}
\ChgEdgeLineWidth{2}
\ChgEdgeLineColor{red}
\ChgEdgeLineStyle{solid}
\ArcL[.5]{1}{8}{}
\ArcL[.5]{2}{1}{}
\ArcL[.5]{3}{2}{}
\ArcL[.5]{7}{6}{}
\ArcL[.5]{8}{7}{}
\EdgeL[.5]{10}{9}{}
\EdgeL[.5]{9}{1}{}
\EdgeL[.5]{12}{11}{}
\EdgeL[.5]{11}{3}{}
\EdgeL[.5]{14}{13}{}
\EdgeL[.5]{13}{3}{}
\EdgeL[.5]{17}{6}{}
\EdgeL[.5]{19}{18}{}
\EdgeL[.5]{18}{7}{}
\EdgeL[.5]{20}{4}{}
\EdgeL[.5]{22}{17}{}
\EdgeL[.5]{21}{3}{}
\ArcL[.5]{5}{4}{}
\VArcR[.5]{arcangle=70,ncurv=0.9}{6}{14}{}
\VArcR[.5]{arcangle=50}{4}{12}{}
\ChgEdgeLineColor{blue}
\ChgEdgeLineWidth{1}
\ChgEdgeLineStyle{dashed}
\ArcL[.5]{4}{3}{}
\ArcL[.5]{6}{5}{}
}
\end{VCPicture}%
        }
        \caption{The picture on the left illustrates Case
          1.2.2. The two edges $(t_1,b_1,p_1)$, $(t_2,b_2,p_2)$ are of
          type 2. The height of the tree $\T'$ obtained after flipping
          the edge $(t_1,b_1,p_1)$ and the red edge going out of
          $t_1$, is equal to $\ell = 2$. In this case, we also flip
          the edge $(t_2,b_2,p_2)$ and the red edge going out of
          $t_2$. We get a unique maximal tree rooted at $r$ in the
          same cluster.  The picture on the right gives the resulting
          cluster.}
\label{figure.cluster3}
\end{figure}

\item Case 1.3. In this case both edges $e_1,e_2$ have type 2 and $t_1
  = t_2$.  We denote by $s_1$ (\resp $s_2$) the child of $r$ ancestor
  of $p_1$ (\resp $p_2$).  We denote by $\T_0$ the tree rooted at $r$
  obtained by the potential flip of $(t_1,b_1,p_1)$ and the red edge
  going out of $t_1$, keeping only $r$ and the subtree rooted at the
  child $s_0$. The nodes of the tree $\T_0$ rooted at $r$ are
  represented in salmon in the left part of
  Fig.~\ref{figure.cluster4}.  This step again needs a computation of
  the height of $\T_0$ explained in the complexity issue. Case 1.3
  occurs when $\rho > 1$, $k=1$ and $t_1=r$.  In the particular case
  where the length of $C$ is $1$, the tree $\T_0$ is reduced to the node
  $r$ (it corresponds to the Case 1.3.2 below).
\begin{itemize}
\item Case 1.3.1. If the height of $\T_0$ is greater than the height of $\T$,
  we flip $(t_1,b_1,p_1)$ and the red edge going out of
  $t_1$. The equivalent automaton satisfies
  Condition~$\C$. 
\item Case 1.3.2. If the height of $\T_0$ is less than the height of
  $\T$, we flip $(t_1,b_1,p_1)$ and the red edge going
  out of $t_1$. We then call again the procedure
  \textsc{FlipEdges}$(\A,r)$ with this new red cycle. This time the
  (new) tree $\T_0$ has the same height as $\T$.  Hence this
  call is done at most one time for a given maximal root $r$. 
\item Case 1.3.3. Finally, we consider the case where the heights of $\T$ and
$\T_0$ are equal (see the left part of Fig.~\ref{figure.cluster4}).
\begin{itemize}
\item Case 1.3.3.1.
If the set of outgoing edges of $s_0$ is a bunch and there is a state
$s_i \in S(r)$ whose set of outgoing edges is also a
bunch, we get a trivial stable pair $(s_0,s_i)$. 
\item Case 1.3.3.2.  If the set of outgoing edges of $s_0$ is a bunch
  and, for any state $s \in S(r)$, the set of outgoing edges of $s$ is
  not a bunch (as in the left part of Fig.~\ref{figure.cluster4}), we
  flip $(t_1,b_1,p_1)$ and the red edge going out of $t_1$. The (new)
  tree $\T_0$ (obtained by the potential flip of $(t_2,b_2,p_2)$ and
  the red edge going out of $t_1$, keeping only $r$ and the subtree
  rooted at the child $s_1$) has the same height as $\T$.  We then
  call again the procedure \textsc{FlipEdges}$(\A,r)$ with this new
  red cycle. This time the height of the new tree $\T_0$ is still
  equal to the height of $\T$ and the set of outgoing edges of the
  predecessor of $r$ on the cycle is not a bunch. This call is thus
  performed at most one time.
\item Case 1.3.3.3.
If the set of outgoing edges of $s_0$ is not a bunch, let
$(s_0,b_0,q_0)$ be a $b$-edge going out of $s_0$ with $q_0 \neq r$.
If $q_0$ does not belong to $\T$, we get an equivalent automaton
satisfying Condition~$\C$ by flipping $(s_0,b_0,q_0)$ and the red edge
going out of $s_0$. If $q_0$ belongs to $\T$, we flip $(s_0,b_0,q_0)$
and the red edge going out of $s_0$. We also flip $(t_1,b_1,p_1)$ and
the red edge going out of $t_1$ if $q_0$ is not a descendant of $s_1$,
or $(t_1,b_2,p_2)$ and the red edge going out of $t_1$, in the opposite
case. Note that $s_0 \neq t_1$ since the height of $\T_0$ is equal to
the non-null height of $\T$. We get an equivalent automaton satisfying
Condition~$\C$ (see the right part of Fig.~\ref{figure.cluster5}).
\end{itemize}
\end{itemize}
\end{itemize}

\begin{figure}[htbp]
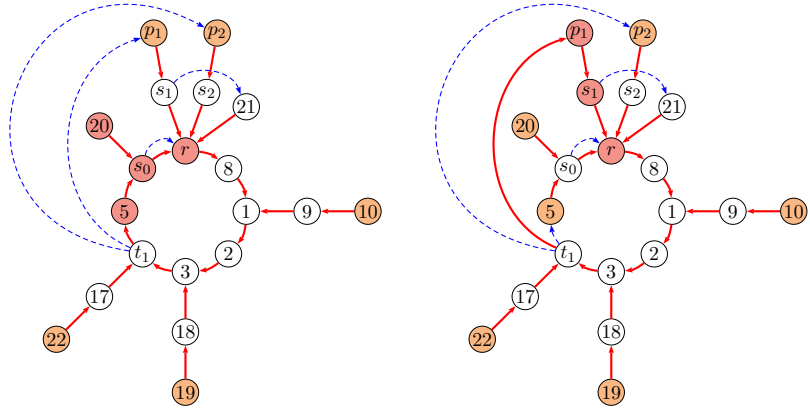

    \centering
\VCDraw{%
\begin{VCPicture}{(0,-1)(12,12)}
\MediumState
\VCPut{(0,6)}{
\State[1]{(2,0)}{1}%
\State[9]{(4,0)}{9}%
\ChgStateFillColor{Apricot}
\State[10]{(6,0)}{10}%
 \RstStateFillColor
\VCPut[45]{(0,0)}{%
    \State[8]{(2,0)}{2}
}
\VCPut[90]{(0,0)}{%
\ChgStateFillColor{Salmon}
    \State[r]{(2,0)}{3}
 \RstStateFillColor
}
\VCPut[60]{(0,0)}{%
    \State[21]{(4,0)}{21}%
}
\VCPut[80]{(0,0)}{%
    \State[s_2]{(4,0)}{11}%
    \ChgStateFillColor{Apricot}
    \State[p_2]{(6,0)}{12}%
    \RstStateFillColor
}
\VCPut[100]{(0,0)}{%
    \State[s_1]{(4,0)}{13}%
    \ChgStateFillColor{Apricot}
    \State[p_1]{(6,0)}{14}%
    \RstStateFillColor  
}
\VCPut[135]{(0,0)}{%
\ChgStateFillColor{Salmon}
    \State[s_0]{(2,0)}{4}
    \State[20]{(4,0)}{20}%
    \RstStateFillColor 
}
\VCPut[180]{(0,0)}{%
\ChgStateFillColor{Salmon}
    \State[5]{(2,0)}{5}
 \RstStateFillColor 
}
\VCPut[225]{(0,0)}{%
    \State[t_1]{(2,0)}{6}
    \State[17]{(4,0)}{17}
    \ChgStateFillColor{Apricot}   
     \State[22]{(6,0)}{22}
    \RstStateFillColor
}
\VCPut[270]{(0,0)}{%
    \State[3]{(2,0)}{7}
    \State[18]{(4,0)}{18}%
    \ChgStateFillColor{Apricot}
    \State[19]{(6,0)}{19}%
    \RstStateFillColor 
}
\VCPut[315]{(0,0)}{%
    \State[2]{(2,0)}{8}}
\ChgEdgeLineWidth{2}
\ChgEdgeLineColor{red}
\ChgEdgeLineStyle{solid}
\ArcL[.5]{1}{8}{}
\ArcL[.5]{2}{1}{}
\ArcL[.5]{3}{2}{}
\ArcL[.5]{4}{3}{}
\ArcL[.5]{5}{4}{}
\ArcL[.5]{6}{5}{}
\ArcL[.5]{7}{6}{}
\ArcL[.5]{8}{7}{}
\EdgeL[.5]{10}{9}{}
\EdgeL[.5]{9}{1}{}
\EdgeL[.5]{12}{11}{}
\EdgeL[.5]{11}{3}{}
\EdgeL[.5]{14}{13}{}
\EdgeL[.5]{13}{3}{}
\EdgeL[.5]{17}{6}{}
\EdgeL[.5]{19}{18}{}
\EdgeL[.5]{18}{7}{}
\EdgeL[.5]{20}{4}{}
\EdgeL[.5]{22}{17}{}
\EdgeL[.5]{21}{3}{}
\ChgEdgeLineColor{blue}
\ChgEdgeLineWidth{1}
\ChgEdgeLineStyle{dashed}
\VArcR[.5]{arcangle=70,ncurv=0.9}{6}{14}{}
\VArcR[.5]{arcangle=100,ncurv=1.8}{6}{12}{}
\VArcR[.5]{arcangle=60}{4}{3}{}
\VArcR[.5]{arcangle=60}{13}{21}{}
}
\VCPut{(14,6)}{

\State[1]{(2,0)}{1}%
\State[9]{(4,0)}{9}%
    \ChgStateFillColor{Apricot}
\State[10]{(6,0)}{10}%
    \RstStateFillColor 
\VCPut[45]{(0,0)}{%
    \State[8]{(2,0)}{2}
}
\VCPut[90]{(0,0)}{%
\ChgStateFillColor{Salmon}
    \State[r]{(2,0)}{3}
    \RstStateFillColor 
}
\VCPut[60]{(0,0)}{%
    \State[21]{(4,0)}{21}%
}
\VCPut[80]{(0,0)}{%
    \State[s_2]{(4,0)}{11}%
    \ChgStateFillColor{Apricot}
    \State[p_2]{(6,0)}{12}%
    \RstStateFillColor 
}
\VCPut[100]{(0,0)}{%
\ChgStateFillColor{Salmon}
    \State[s_1]{(4,0)}{13}%
    \State[p_1]{(6,0)}{14}%
    \RstStateFillColor 
}
\VCPut[135]{(0,0)}{%
    \State[s_0]{(2,0)}{4}
    \ChgStateFillColor{Apricot}
    \State[20]{(4,0)}{20}%
    \RstStateFillColor 
}
\VCPut[180]{(0,0)}{%
    \ChgStateFillColor{Apricot}  
     \State[5]{(2,0)}{5}  
    \RstStateFillColor 
}
\VCPut[225]{(0,0)}{%
    \State[t_1]{(2,0)}{6}
    \State[17]{(4,0)}{17}
    \ChgStateFillColor{Apricot}
     \State[22]{(6,0)}{22}
    \RstStateFillColor 
}
\VCPut[270]{(0,0)}{%
    \State[3]{(2,0)}{7}
    \State[18]{(4,0)}{18}%
    \ChgStateFillColor{Apricot}
    \State[19]{(6,0)}{19}%
    \RstStateFillColor 
}
\VCPut[315]{(0,0)}{%
    \State[2]{(2,0)}{8}}
\ChgEdgeLineWidth{2}
\ChgEdgeLineColor{red}
\ChgEdgeLineStyle{solid}
\ArcL[.5]{1}{8}{}
\ArcL[.5]{2}{1}{}
\ArcL[.5]{3}{2}{}
\ArcL[.5]{7}{6}{}
\ArcL[.5]{8}{7}{}
\EdgeL[.5]{10}{9}{}
\EdgeL[.5]{9}{1}{}
\EdgeL[.5]{12}{11}{}
\EdgeL[.5]{11}{3}{}
\EdgeL[.5]{14}{13}{}
\EdgeL[.5]{13}{3}{}
\EdgeL[.5]{17}{6}{}
\EdgeL[.5]{19}{18}{}
\EdgeL[.5]{18}{7}{}
\EdgeL[.5]{20}{4}{}
\EdgeL[.5]{22}{17}{}
\EdgeL[.5]{21}{3}{}
\ArcL[.5]{5}{4}{}
\VArcR[.5]{arcangle=70,ncurv=0.9}{6}{14}{}
\ArcL[.5]{4}{3}{}
\ChgEdgeLineColor{blue}
\ChgEdgeLineWidth{1}
\ChgEdgeLineStyle{dashed}
\ArcL[.5]{6}{5}{}
\VArcR[.5]{arcangle=100,ncurv=1.8}{6}{12}{}
\VArcR[.5]{arcangle=60}{4}{3}{}
\VArcR[.5]{arcangle=60}{13}{21}{}
}
\end{VCPicture}%
        }
        \caption{ The picture on the left illustrates Case 1.3.3.2 of
          the main treatment. The two edges $(t_1,b_1,p_1)$ and
          $(t_1,b_2,p_2)$ are of type 2.  Let $\T_0$ be the tree
          rooted at $r$ obtained by the potential flip of
          $(t_1,b_1,p_1)$ and the red edge going out of $t_1$, keeping
          only $r$ and the subtree rooted at the child $s_0$. The
          nodes of the tree $\T_0$ rooted at $r$ are represented in
          salmon in the left part of the figure. The state $s_0$ is a
          bunch. After flipping the edge $(t_1,b_1,p_1)$ and the red
          edge going out of $t_1$, we get the automaton pictured in
          the right part of the figure.  The tree $\T'_0$ is now tree
          rooted at $r$ obtained by the potential flip of
          $(t_1,b_2,p_2)$ and the red edge going out of $t_1$, keeping
          only $r$ and the subtree rooted at the child $s_1$. Its
          states are colored in salmon. The height of $T'_0$ is 2.}
\label{figure.cluster4}
\end{figure}


\begin{figure}[htbp]
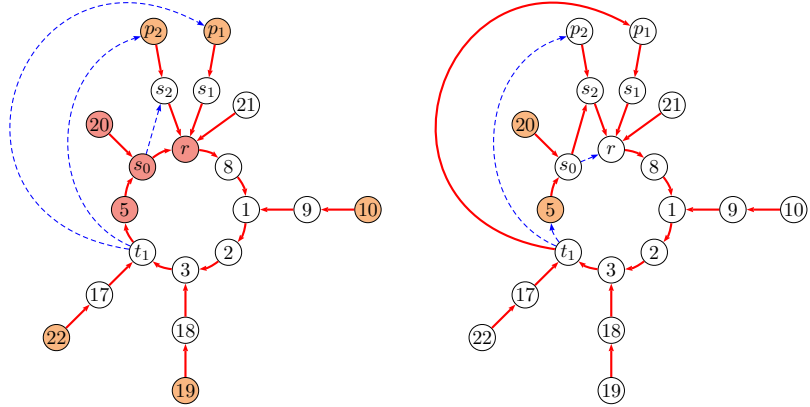

    \centering
\VCDraw{%
\begin{VCPicture}{(0,-1)(12,12)}
\MediumState
\VCPut{(0,6)}{
\State[1]{(2,0)}{1}%
\State[9]{(4,0)}{9}%
\ChgStateFillColor{Apricot}
\State[10]{(6,0)}{10}%
 \RstStateFillColor
\VCPut[45]{(0,0)}{%
    \State[8]{(2,0)}{2}
}
\VCPut[90]{(0,0)}{%
\ChgStateFillColor{Salmon}
    \State[r]{(2,0)}{3}
 \RstStateFillColor
}
\VCPut[60]{(0,0)}{%
    \State[21]{(4,0)}{21}%
}
\VCPut[80]{(0,0)}{%
    \State[s_1]{(4,0)}{11}%
    \ChgStateFillColor{Apricot}
    \State[p_1]{(6,0)}{12}%
    \RstStateFillColor
}
\VCPut[100]{(0,0)}{%
    \State[s_2]{(4,0)}{13}%
    \ChgStateFillColor{Apricot}
    \State[p_2]{(6,0)}{14}%
    \RstStateFillColor  
}
\VCPut[135]{(0,0)}{%
\ChgStateFillColor{Salmon}
    \State[s_0]{(2,0)}{4}
    \State[20]{(4,0)}{20}%
    \RstStateFillColor 
}
\VCPut[180]{(0,0)}{%
\ChgStateFillColor{Salmon}
    \State[5]{(2,0)}{5}
 \RstStateFillColor 
}
\VCPut[225]{(0,0)}{%
    \State[t_1]{(2,0)}{6}
    \State[17]{(4,0)}{17}
    \ChgStateFillColor{Apricot}   
     \State[22]{(6,0)}{22}
    \RstStateFillColor
}
\VCPut[270]{(0,0)}{%
    \State[3]{(2,0)}{7}
    \State[18]{(4,0)}{18}%
    \ChgStateFillColor{Apricot}
    \State[19]{(6,0)}{19}%
    \RstStateFillColor 
}
\VCPut[315]{(0,0)}{%
    \State[2]{(2,0)}{8}}
\ChgEdgeLineWidth{2}
\ChgEdgeLineColor{red}
\ChgEdgeLineStyle{solid}
\ArcL[.5]{1}{8}{}
\ArcL[.5]{2}{1}{}
\ArcL[.5]{3}{2}{}
\ArcL[.5]{4}{3}{}
\ArcL[.5]{5}{4}{}
\ArcL[.5]{6}{5}{}
\ArcL[.5]{7}{6}{}
\ArcL[.5]{8}{7}{}
\EdgeL[.5]{10}{9}{}
\EdgeL[.5]{9}{1}{}
\EdgeL[.5]{12}{11}{}
\EdgeL[.5]{11}{3}{}
\EdgeL[.5]{14}{13}{}
\EdgeL[.5]{13}{3}{}
\EdgeL[.5]{17}{6}{}
\EdgeL[.5]{19}{18}{}
\EdgeL[.5]{18}{7}{}
\EdgeL[.5]{20}{4}{}
\EdgeL[.5]{22}{17}{}
\EdgeL[.5]{21}{3}{}
\ChgEdgeLineColor{blue}
\ChgEdgeLineWidth{1}
\ChgEdgeLineStyle{dashed}
\VArcR[.5]{arcangle=70,ncurv=0.9}{6}{14}{}
\VArcR[.5]{arcangle=100,ncurv=1.8}{6}{12}{}
\EdgeL[.5]{4}{13}{}
}
\VCPut{(14,6)}{

\State[1]{(2,0)}{1}%
\State[9]{(4,0)}{9}%
\State[10]{(6,0)}{10}%
\VCPut[45]{(0,0)}{%
    \State[8]{(2,0)}{2}
}
\VCPut[90]{(0,0)}{%
    \State[r]{(2,0)}{3}
}
\VCPut[60]{(0,0)}{%
    \State[21]{(4,0)}{21}%
}
\VCPut[80]{(0,0)}{%
    \State[s_1]{(4,0)}{11}%
    \State[p_1]{(6,0)}{12}%
}
\VCPut[100]{(0,0)}{%
    \State[s_2]{(4,0)}{13}%
    \State[p_2]{(6,0)}{14}%
}
\VCPut[135]{(0,0)}{%
    \State[s_0]{(2,0)}{4}
    \ChgStateFillColor{Apricot}
    \State[20]{(4,0)}{20}%
    \RstStateFillColor 
}
\VCPut[180]{(0,0)}{%
    \ChgStateFillColor{Apricot}  
     \State[5]{(2,0)}{5}  
    \RstStateFillColor 
}
\VCPut[225]{(0,0)}{%
    \State[t_1]{(2,0)}{6}
    \State[17]{(4,0)}{17}
     \State[22]{(6,0)}{22}
}
\VCPut[270]{(0,0)}{%
    \State[3]{(2,0)}{7}
    \State[18]{(4,0)}{18}%
    \State[19]{(6,0)}{19}%
}
\VCPut[315]{(0,0)}{%
    \State[2]{(2,0)}{8}}
\ChgEdgeLineWidth{2}
\ChgEdgeLineColor{red}
\ChgEdgeLineStyle{solid}
\ArcL[.5]{1}{8}{}
\ArcL[.5]{2}{1}{}
\ArcL[.5]{3}{2}{}
\ArcL[.5]{7}{6}{}
\ArcL[.5]{8}{7}{}
\EdgeL[.5]{10}{9}{}
\EdgeL[.5]{9}{1}{}
\EdgeL[.5]{12}{11}{}
\EdgeL[.5]{11}{3}{}
\EdgeL[.5]{14}{13}{}
\EdgeL[.5]{13}{3}{}
\EdgeL[.5]{17}{6}{}
\EdgeL[.5]{19}{18}{}
\EdgeL[.5]{18}{7}{}
\EdgeL[.5]{20}{4}{}
\EdgeL[.5]{22}{17}{}
\EdgeL[.5]{21}{3}{}
\ArcL[.5]{5}{4}{}
\EdgeL[.5]{4}{13}{}
\VArcR[.5]{arcangle=100,ncurv=1.8}{6}{12}{}
\ChgEdgeLineColor{blue}
\ChgEdgeLineWidth{1}
\ChgEdgeLineStyle{dashed}
\ArcL[.5]{6}{5}{}
\VArcR[.5]{arcangle=70,ncurv=0.9}{6}{14}{}
\EdgeL[.5]{4}{3}{}
}
\end{VCPicture}%
        }
        \caption{ The picture on the left illustrates Case 1.3.3.3. The
          two edges $(t_1,b_1,p_1)$ and $(t_1,b_2,p_2)$ are of type 2.
          Let $\T_0$ be the tree rooted at $r$ obtained by the
          potential flip of $(t_1,b_1,p_1)$ and the red edge going out
          of $t_1$, keeping only $r$ and the subtree rooted at the
          child $s_0$. The nodes of the tree $\T_0$ rooted at $r$ are
          represented in salmon in the left part of the figure.  The
          state $s_0$ is not a bunch: it has a $b$-edge
          $(s_0,b_0,q_0)$ with $q_0=s_2$. After flipping the edge
          $(t_1,b_1,p_1)$ and the red edge going out of $t_1$, and
          flipping $(s_0,b_0,q_0)$ and the red edge going out of
          $s_0$, we get a unique maximal tree rooted at $r$ in the
          same cluster (see the right part of the figure). }
\label{figure.cluster5}
\end{figure}

\item Case 2. We now come to the case where the value returned by
  \textsc{FindEdges}$(r)$ is a pair $(2,e)$ with $e=(t_1,b_1,p_1)$ of
  type 1 or 2, and $\T$ is modified in such a way that 
 $r$ has a unique maximal child, \ie $\rho = 1$.
\begin{itemize}
\item Case 2.1.  If $(t_1,b_1,p_1)$ has type 1, we flip the edge
  $(t_1,b_1,p_1)$ and the red edge going out of~$t_1$. We get an
  equivalent automaton satisfying Condition~$\C$.
\item Case 2.2.  If $(t_1,b_1,p_1)$ has type 2, we denote by $\T_0$
  the tree rooted at $r$ obtained by the potential flip of
  $(t_1,b_1,p_1)$ and the red edge going out of $t_1$, keeping only
  $r$ and the subtree rooted at the child $s_0$.  Case 2.2 occurs
  when $\rho=1$, $k=1$ and $t_1=r$. In the particular case where the
  length of $C$ is $1$, $\T_0$ is reduced to the node $r$ which
  corresponds to the Case 2.2.2 below.
\begin{itemize}
\item Case 2.2.1.  If the height of $\T_0$ is greater than the height
  of $\T$, we do the flip and the equivalent automaton satisfies
  Condition~$\C$.  
\item Case 2.2.2.  If the height of $\T_0$ is less than the height of
  $\T$, we do not do the flip, and return the automaton together with
  the edge $(t_1,b_1,p_1)$. Note that a possible future flip of
  $(t_1,b_1,p_1)$ and the red edge starting at $t_1$ makes the root
  $r$ not maximal anymore.
\item Case 2.2.3.  
We now come to the case where the height
  of $\T_0$ is equal to the height of $\T$.  
\begin{itemize}
\item Case 2.2.3.1.  If the set of outgoing edges of $s_0$ and $s_1$
  are bunches, there is a trivial stable pair $(s_0,s_1)$.
\item Case 2.2.3.2.
If the set of outgoing edges of $s_0$ is a bunch and
  the set of outgoing edges of $s_1$ is not a bunch (see the left
  part of Fig.~\ref{figure.cluster6}), we flip the edge
  $(t_1,b_1,p_1)$ and the red edge going out of $t_1$.  We then call
  the procedure \textsc{FlipEdges}$(\A,r)$ with this new red
  cycle. The root $r$ has now a unique child ($s_1$) ancestor of
  maximal state whose set of outgoing edges is a bunch (see the right
  part of Fig.~\ref{figure.cluster6}).  This call is thus performed
  at most one time.  
\item Case 2.2.3.3.
Finally, if $s_0$ is a not a bunch, let
  $(s_0,b_0,q_0)$ be a $b$-edge with $q_0 \neq r$.  If $q_0$ does not
  belong to $\T$ we flip the edge $(s_0,b_0,q_0)$ and the red edge
  going out of $s_0$.  The equivalent automaton satisfies
  Condition~$\C$. It $q_0$ belongs to $\T$ and is not a descendant of
  $s_1$, we flip the edge $(t_1,b_1,p_1)$ and the red edge going out
  of $t_1$, and we also flip the edge $(s_0,b_0,q_0)$ and the red edge
  going out of $s_0$.  The equivalent automaton satisfies
  Condition~$\C$.  If $q_0$ belongs to $\T$ and is a descendant of
  $s_1$, we return the automaton together with the edge
  $(s_0,b_0,q_0)$.
\end{itemize}
\end{itemize}
\end{itemize}


\begin{figure}[htbp]
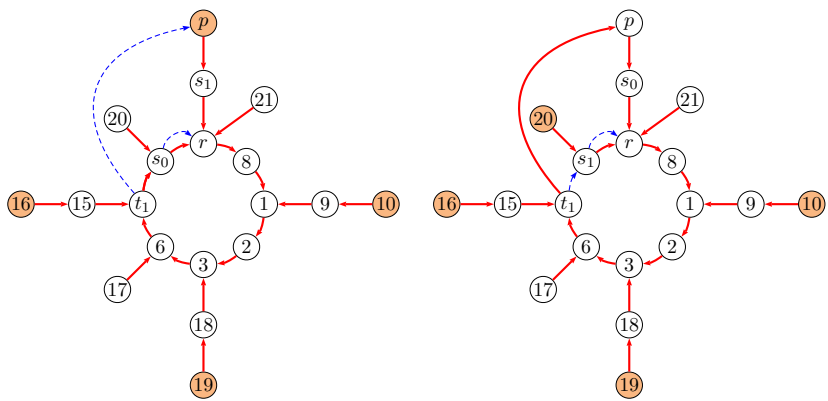

    \centering
\VCDraw{%
\begin{VCPicture}{(0,-1)(12,12)}
\MediumState
\VCPut{(0,6)}{
\State[1]{(2,0)}{1}%
\State[9]{(4,0)}{9}%
\ChgStateFillColor{Apricot}
\State[10]{(6,0)}{10}%
 \RstStateFillColor
\VCPut[45]{(0,0)}{%
    \State[8]{(2,0)}{2}
}
\VCPut[90]{(0,0)}{%
    \State[r]{(2,0)}{3}
}
\VCPut[60]{(0,0)}{%
    \State[21]{(4,0)}{21}%
}
\VCPut[90]{(0,0)}{%
    \State[s_1]{(4,0)}{13}%
    \ChgStateFillColor{Apricot}
    \State[p]{(6,0)}{14}%
    \RstStateFillColor
}
\VCPut[135]{(0,0)}{%
    \State[s_0]{(2,0)}{4}
    \State[20]{(4,0)}{new}}
\VCPut[180]{(0,0)}{%
    \State[t_1]{(2,0)}{5}
    \State[15]{(4,0)}{15}%
    \ChgStateFillColor{Apricot}
    \State[16]{(6,0)}{16}%
    \RstStateFillColor
}
\VCPut[225]{(0,0)}{%
    \State[6]{(2,0)}{6}
    \State[17]{(4,0)}{17}
}
\VCPut[270]{(0,0)}{%
    \State[3]{(2,0)}{7}
    \State[18]{(4,0)}{18}%
   \ChgStateFillColor{Apricot}
    \State[19]{(6,0)}{19}%
    \RstStateFillColor
}
\VCPut[315]{(0,0)}{%
    \State[2]{(2,0)}{8}}
\ChgEdgeLineWidth{2}
\ChgEdgeLineColor{red}
\ChgEdgeLineStyle{solid}
\ArcL[.5]{1}{8}{}
\ArcL[.5]{2}{1}{}
\ArcL[.5]{3}{2}{}
\ArcL[.5]{4}{3}{}
\ArcL[.5]{7}{6}{}
\ArcL[.5]{8}{7}{}
\EdgeL[.5]{10}{9}{}
\EdgeL[.5]{9}{1}{}
\EdgeL[.5]{14}{13}{}
\EdgeL[.5]{16}{15}{}
\EdgeL[.5]{15}{5}{}
\EdgeL[.5]{17}{6}{}
\EdgeL[.5]{19}{18}{}
\EdgeL[.5]{18}{7}{}
\ArcL[.5]{6}{5}{}
\ArcL[.5]{5}{4}{}
\EdgeL[.5]{21}{3}{}
\EdgeL[.5]{13}{3}{}
\ArcL[.5]{5}{4}{}
\EdgeL[.5]{new}{4}{}
\ChgEdgeLineColor{blue}
\ChgEdgeLineWidth{1}
\ChgEdgeLineStyle{dashed}
\VArcR[.5]{arcangle=60,ncurv=1.2}{5}{14}{}
\VArcR[.5]{arcangle=60}{4}{3}{}
}
\VCPut{(14,6)}{
\State[1]{(2,0)}{1}%
\State[9]{(4,0)}{9}%
\ChgStateFillColor{Apricot}
\State[10]{(6,0)}{10}%
 \RstStateFillColor
\VCPut[45]{(0,0)}{%
    \State[8]{(2,0)}{2}
}
\VCPut[90]{(0,0)}{%
    \State[r]{(2,0)}{3}
}
\VCPut[60]{(0,0)}{%
    \State[21]{(4,0)}{21}%
}
\VCPut[90]{(0,0)}{%
    \State[s_0]{(4,0)}{13}%
    \State[p]{(6,0)}{14}%
}
\VCPut[135]{(0,0)}{%
    \State[s_1]{(2,0)}{4}
     \ChgStateFillColor{Apricot}
    \State[20]{(4,0)}{new}
      \RstStateFillColor
}
\VCPut[180]{(0,0)}{%
    \State[t_1]{(2,0)}{5}
    \State[15]{(4,0)}{15}%
    \ChgStateFillColor{Apricot}
    \State[16]{(6,0)}{16}%
    \RstStateFillColor
}
\VCPut[225]{(0,0)}{%
    \State[6]{(2,0)}{6}
    \State[17]{(4,0)}{17}
}
\VCPut[270]{(0,0)}{%
    \State[3]{(2,0)}{7}
    \State[18]{(4,0)}{18}%
   \ChgStateFillColor{Apricot}
    \State[19]{(6,0)}{19}%
    \RstStateFillColor
}
\VCPut[315]{(0,0)}{%
    \State[2]{(2,0)}{8}}
\ChgEdgeLineWidth{2}
\ChgEdgeLineColor{red}
\ChgEdgeLineStyle{solid}
\ArcL[.5]{1}{8}{}
\ArcL[.5]{2}{1}{}
\ArcL[.5]{3}{2}{}
\ArcL[.5]{4}{3}{}
\ArcL[.5]{7}{6}{}
\ArcL[.5]{8}{7}{}
\EdgeL[.5]{10}{9}{}
\EdgeL[.5]{9}{1}{}
\EdgeL[.5]{14}{13}{}
\EdgeL[.5]{16}{15}{}
\EdgeL[.5]{15}{5}{}
\EdgeL[.5]{17}{6}{}
\EdgeL[.5]{19}{18}{}
\EdgeL[.5]{18}{7}{}
\ArcL[.5]{6}{5}{}
\EdgeL[.5]{21}{3}{}
\EdgeL[.5]{13}{3}{}
\EdgeL[.5]{new}{4}{}
\VArcR[.5]{arcangle=60,ncurv=1.2}{5}{14}{}
\ChgEdgeLineColor{blue}
\ChgEdgeLineWidth{1}
\ChgEdgeLineStyle{dashed}
\ArcL[.5]{5}{4}{}
\VArcR[.5]{arcangle=60}{4}{3}{}
}
\end{VCPicture}%
        }
        \caption{The
          picture on the left illustrates Case 2.2.3.2 of the main treatment.
          The edge $(t_1,b_1,p_1)$ has type 2.  After flipping the edge
          $(t_1,b_1,p_1)$ and the red edge going out of $t_1$, we get the
          automaton on the right part of the figure. The root $r$ has a new single child
          $s_1$ ancestor of a maximal state, whose set of outgoing
          edges is a bunch. The new tree rooted at $r$ has here the
          same level $\ell =2$ as before and \textsc{FlipEdges}$(\A,r)$ is
          called a second and last time.}
\label{figure.cluster6}
\end{figure}
\item Case 3. If the value returned by \textsc{FindEdges}$(r)$ is an
  edge $(t_1,b_1,p_1)$ of type 3 and $t_1$ is an ancestor of all maximal
  nodes of $\T$ the procedure \textsc{FlipEdges}$(\A,r)$ returns this
  edge.
\end{itemize}

After running $\textsc{FlipEdges}(\A,r)$ on all maximal roots, we
get either an automaton with a stable pair, or an automaton where each
cluster fulfills the following conditions.
\begin{itemize}
\item the root $r$ of each maximal tree has a unique maximal child;
\item for each maximal root $r$, there is an edge $(t_r,b_r,p_r)$ such
  that the potential flip of $(t_r,b_r,p_r)$ and the red edge starting at 
  $t_r$ makes the root $r$ not maximal anymore.
\end{itemize}
If the latter case, we flip the blue edge $(t_r,b_r,p_r)$ and the red
one starting at $t_r$ for all maximal roots $r$ but one. We get an
equivalent automaton which satisfies Condition~$\C$ as is shown in
Fig.~\ref{figure.cluster7}. The pseudocode for this final treatment
is given in procedure \textsc{FindStablePair}.

\begin{small}
\begin{codebox} 
\Procname{$\proc{FindStablePair }(\text{automaton } \A$)} 
\li  \If the maximal level $\ell = 0$
\li  \Then  \Return $\proc{LevelZeroFlipEdges}(\A)$
\li   \Else \For each maximal root $r$
\li        \Do $\A, S  \gets \proc{FlipEdges}(\A,r)$
\li         \If $S$ is a (stable) pair of states $(s,t)$
\li             \Then \Return $\A$, $(s,t)$
\li             \Else ($S$ is a $b$-edge $(t_r,b_r,p_r)$) set $e(r) = S$
             \End
           \End
\li      \For each maximal root $r \neq r_0$ 
\li        \Do flip the edge $e(r)$ and the red edge starting at $t_r$
           \End
\label{li:FindStablePair-fin}
\li      $s \gets$ \proc{GetPredecessor}$(r_0)$
\li      $t \gets$ the child of $r_0$ ancestor of $p_{r_0}$
\label{li:FindStablePair-fin2}
\li     \Return $\A$, $(s,t)$ 
    \End
\end{codebox}
\end{small}

\begin{figure}[htbp]
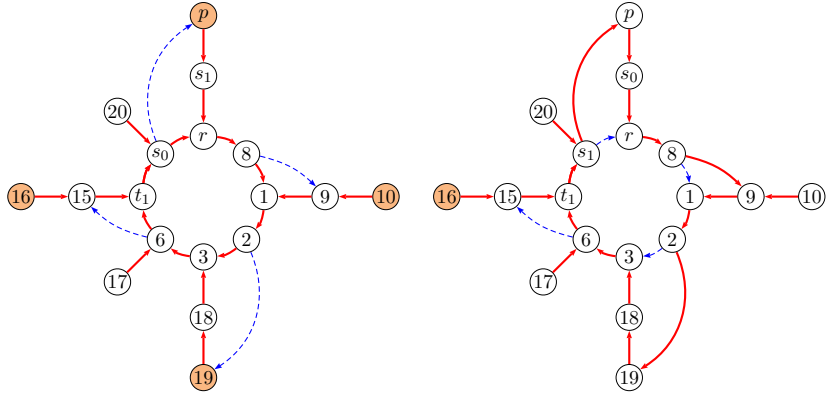

    \centering
\VCDraw{%
\begin{VCPicture}{(0,-1)(12,12)}
\MediumState
\VCPut{(0,6)}{
\State[1]{(2,0)}{1}%
\State[9]{(4,0)}{9}%
\ChgStateFillColor{Apricot}
\State[10]{(6,0)}{10}%
 \RstStateFillColor
\VCPut[45]{(0,0)}{%
    \State[8]{(2,0)}{2}
}
\VCPut[90]{(0,0)}{%
    \State[r]{(2,0)}{3}
}
\VCPut[90]{(0,0)}{%
    \State[s_1]{(4,0)}{13}%
    \ChgStateFillColor{Apricot}
    \State[p]{(6,0)}{14}%
    \RstStateFillColor
}
\VCPut[135]{(0,0)}{%
    \State[s_0]{(2,0)}{4}
    \State[20]{(4,0)}{new}}
\VCPut[180]{(0,0)}{%
    \State[t_1]{(2,0)}{5}
    \State[15]{(4,0)}{15}%
    \ChgStateFillColor{Apricot}
    \State[16]{(6,0)}{16}%
    \RstStateFillColor
}
\VCPut[225]{(0,0)}{%
    \State[6]{(2,0)}{6}
    \State[17]{(4,0)}{17}
}
\VCPut[270]{(0,0)}{%
    \State[3]{(2,0)}{7}
    \State[18]{(4,0)}{18}%
   \ChgStateFillColor{Apricot}
    \State[19]{(6,0)}{19}%
    \RstStateFillColor
}
\VCPut[315]{(0,0)}{%
    \State[2]{(2,0)}{8}}
\ChgEdgeLineWidth{2}
\ChgEdgeLineColor{red}
\ChgEdgeLineStyle{solid}
\ArcL[.5]{1}{8}{}
\ArcL[.5]{2}{1}{}
\ArcL[.5]{3}{2}{}
\ArcL[.5]{4}{3}{}
\ArcL[.5]{7}{6}{}
\ArcL[.5]{8}{7}{}
\EdgeL[.5]{10}{9}{}
\EdgeL[.5]{9}{1}{}
\EdgeL[.5]{14}{13}{}
\EdgeL[.5]{16}{15}{}
\EdgeL[.5]{15}{5}{}
\EdgeL[.5]{17}{6}{}
\EdgeL[.5]{19}{18}{}
\EdgeL[.5]{18}{7}{}
\ArcL[.5]{6}{5}{}
\ArcL[.5]{5}{4}{}
\EdgeL[.5]{13}{3}{}
\ArcL[.5]{5}{4}{}
\EdgeL[.5]{new}{4}{}
\ChgEdgeLineColor{blue}
\ChgEdgeLineWidth{1}
\ChgEdgeLineStyle{dashed}
\VArcR[.5]{arcangle=40}{4}{14}{}
\ArcL[.5]{2}{9}{}
\VArcR[.5]{arcangle=40}{8}{19}{}
\ArcL[.5]{6}{15}{}
}
\VCPut{(14,6)}{
\State[1]{(2,0)}{1}%
\State[9]{(4,0)}{9}%
\State[10]{(6,0)}{10}%
\VCPut[45]{(0,0)}{%
    \State[8]{(2,0)}{2}
}
\VCPut[90]{(0,0)}{%
    \State[r]{(2,0)}{3}
}
\VCPut[90]{(0,0)}{%
    \State[s_0]{(4,0)}{13}%
    \State[p]{(6,0)}{14}%
}
\VCPut[135]{(0,0)}{%
    \State[s_1]{(2,0)}{4}
    \State[20]{(4,0)}{new}
}
\VCPut[180]{(0,0)}{%
    \State[t_1]{(2,0)}{5}
    \State[15]{(4,0)}{15}%
    \ChgStateFillColor{Apricot}
    \State[16]{(6,0)}{16}%
    \RstStateFillColor
}
\VCPut[225]{(0,0)}{%
    \State[6]{(2,0)}{6}
    \State[17]{(4,0)}{17}
}
\VCPut[270]{(0,0)}{%
    \State[3]{(2,0)}{7}
    \State[18]{(4,0)}{18}%
    \State[19]{(6,0)}{19}%
}
\VCPut[315]{(0,0)}{%
    \State[2]{(2,0)}{8}}
\ChgEdgeLineWidth{2}
\ChgEdgeLineColor{red}
\ArcL[.5]{1}{8}{}
\ArcL[.5]{3}{2}{}
\ArcL[.5]{7}{6}{}
\EdgeL[.5]{10}{9}{}
\EdgeL[.5]{9}{1}{}
\EdgeL[.5]{14}{13}{}
\EdgeL[.5]{16}{15}{}
\EdgeL[.5]{15}{5}{}
\EdgeL[.5]{17}{6}{}
\EdgeL[.5]{19}{18}{}
\EdgeL[.5]{18}{7}{}
\ArcL[.5]{6}{5}{}
\ArcL[.5]{5}{4}{}
\EdgeL[.5]{13}{3}{}
\EdgeL[.5]{new}{4}{}
\ArcL[.5]{5}{4}{}
\VArcR[.5]{arcangle=40}{4}{14}{}
\ArcL[.5]{2}{9}{}
\VArcR[.5]{arcangle=40}{8}{19}{}
\ChgEdgeLineColor{blue}
\ChgEdgeLineWidth{1}
\ChgEdgeLineStyle{dashed}
\ArcL[.5]{6}{15}{}
\ArcL[.5]{2}{1}{}
\ArcL[.5]{4}{3}{}
\ArcL[.5]{8}{7}{}
}
\end{VCPicture}%
        }
        \caption{ The picture on the left illustrates the case where
          \textsc{FlipEdges}$(\A,r)$ has returned a $b$-edge $e(r)$
          for all maximal roots $r$.  We flip $e(r)$ and the red edge
          starting at the same state for all but one maximal root
          $r$. The new cluster is pictured on the right part of the
          figure. It has a unique maximal tree. By
          Lemma~\ref{lemma.sameTree} the pair $(6,15)$ is stable.}
\label{figure.cluster7}
\end{figure}

\pagebreak
\subsubsection{The auxiliary procedure \textsc{FindEdges}} \label{section.auxiliary}

In this section, we describe the procedure \textsc{FindEdges}$(r)$
which is a preliminary step of the procedure \textsc{FlipEdges}$(r)$.

Let $r$ be a maximal root, $S(r)$ be the set of maximal children of
$r$. For each $s$ in $S(r)$, we choose a maximal state $p$ in the
subtree rooted at $s$ and we denote by $P(r)$ the set of these maximal
states (see Fig. \ref{figure.tree}). Recall that the procedure
\textsc{FindEdges}$(r)$ flips some edges and returns an equivalent
automaton together with one or two edges of the following forms.
\begin{itemize}
\item One edge $e$ of type 0.
\item Two edges $e,f$ of type 1 or 2 ending in distinct states of
  $P(r)$.
\item One edge $e$ of type 1 or 2.  Moreover, in this case, the
  procedure modifies the tree $\T$ in such a way that $r$ has a unique
  maximal child.
\item One edge $e$ of type 3 starting at a
  state which is an ancestor of all maximal nodes of $\T$.
\end{itemize} 

For each maximal child $s$, we denote by $T_s$ the subtree of $T$
rooted at $s$. The procedure \textsc{FindEdges}$(r,s)$ computes a
list $L_{s}$ of $b$-edges $(q,b,p)$, where $p$ is a maximal node of
$T_s$ and $q$ is an ancestor of $p$ in $T$ distinct from $r$. The
starting states $q$ of edges of this list \emph{cover} the maximal
nodes of $T_s$ in the following sense: for each maximal node $p'$ in
$T_s$, there is a unique edge $(t,b,p) \in L_s$ such that $t$ is an
ancestor of $p'$ (see for instance the right part of Fig.
\ref{figure.tree}). The list $L_s$ is computed by scanning at most one
time each node of the tree $T_s$. For each maximal leaf $p$, we
follow the red edges up to $s$ and either find $s$ or an already
scanned node, or find a node with an outgoing $b$-edge ending in
$p$. In the latter case, this edge is added to $L_s$ and we continue
with another maximal leaf. In the case the list $L_s$ does not cover
all maximal nodes of $T_s$, and since the graph of the automaton
is strongly connected, the process finds an edge $(t_s,b_s,p_s)$ where
$p_s$ is a maximal node of $T_s$, of type 0, 1 or 2.

If there is a maximal child $s$ such that an edge $(t_s,b_s,p_s)$ of
type 0 is found, then \textsc{FindEdges}$(r)$ returns this
edge. 

Otherwise, if there are two maximal children $s_1 \neq s_2$ such
that two edges $(t_{s_1},b_{s_1},p_{s_1})$,
$(t_{s_2},b_{s_2},p_{s_2})$ of type 1 or 2 are found, then
\textsc{FindEdges}$(r)$ returns these two edges.  If there is a
maximal child $s_1$ such an edge $e=(t_{s_1},b_{s_1},p_{s_1})$ of type 1
or 2 and covering lists $L_s$ for the other maximal children $s \neq
s_1$ are found, then we perform the following flips. For any maximal
child $s \neq s_1$ and any edge $(t,b,p) \in L_s$, we flip the edge
$(t,b,p)$ and the red edge going out of $t$. We update the data of the
trees attached to the nodes from $p$ to $t$ in the new red cycle
created by the flip. After this transformation the node $r$ has $s_1$ as
unique maximal child. 
The procedure \textsc{FindEdges}$(r)$ returns the edge $e$ of type 
1 or 2  and $r$ has a unique maximal child.  

Finally, if one obtains covering lists for all maximal children, then,
for all these children $s$ but one, say $s_1$, we flip each edge
$(t,b,p) \in L_s$ and the red edge going out of $t$. We also flip all
edges $(t,b,p) \in L_{s_1}$ but one, $(t_1,b_1,p_1)$. We update the
data of the trees attached to the nodes from $p$ to $t$ in the new red
cycle created by each flip.  The procedure \textsc{FindEdges}$(r)$
returns the edge $(t_1,b_1,p_1)$ of type 3. Its starting state $t_1$
is disctinct from $r$ ans is an ancestor of all maximal states of $\T$.


\section{The complexity issue} \label{section.complexity}

In this section, we establish the time and space complexity of our
algorithm. We denote by $k$ the size of the alphabet $A$ and by $n$
the number of states of~$\A$.  Since $\A$ is complete deterministic,
it has exactly $kn$ edges.

\subsection{Data structures and their updating} \label{section.update}

Some data attached to the states is useful to obtain the claimed
complexity.  This data is updated after the computation of each
quotient automaton with the procedure $\textsc{Update}$ with a time
complexity which is linear in the size of the quotient automaton.

The edges of the automaton can be stored in tables indexed by the
states and labels. The updating procedure computes the level of each
state, the root of its tree in its cluster.  It also computes a list
of maximal roots and the predecessor of a state on the cycle.  The
function \textsc{GetPredecessor}$(q)$ returns the predecessor of state
$q$ on its red cycle in constant time.

One computes 
\begin{itemize}
\item for each root of a tree $T$, the height of $T$,
\item for each maximal root, the list of its maximal children,
\item for each maximal child, the list of the maximal nodes belonging
  to the subtree rooted at this child.
\end{itemize}
This data can be moreover updated in time linear in the size of the
tree.

We also maintain an inverse structure of the quotient
automaton. Giving a label $c$ and a state $q$, it gives, for each
letter $c$, an unordered list of states $p$ such that there is an edge
$(p,c,q)$ in the quotient automaton. The procedure
\textsc{Flip}$(e,f)$ exchanges the labels of the two edges
$e=(p,b,q),f=(p,a,q')$. It also updates in the inverse structure the
lists of edges coming in $p$ and $p'$. Its time complexity is thus
upper bounded by the number of edges going out of $p,p'$ or coming in
$p,p'$.

\subsection{Complexity of the algorithm} 

\begin{proposition} \label{proposition.complexity}
The worst-case complexity of \textsc{FindSynchronizedColoring}
applied to an $n$-state aperiodic automaton is $O(kn^2)$
in time, and $O(kn)$ in space.
\end{proposition}

\begin{proof} 
  The complexity of \textsc{FindSynchronizedColoring} is at most $n$ times the
  complexity of the procedures \textsc{Update} and
  \textsc{FindStablePair}.  Indeed, each call in the procedure
  \textsc{Merge} reduces the number of states of the automaton so that
  it is called at most $n-1$ times. Since each of its steps without the
  recursive calls takes a time at most $kn$, the contribution of
  \textsc{Merge} in \textsc{FindSynchronizedColoring} is at most $kn$.  As the
  procedure \textsc{Update} has a time complexity $O(kn)$, we just
  have to show that the time complexity of
  \textsc{FindStablePair} is $O(kn)$.

  Since \textsc{LevelZeroFlipEdges} contains only one \textsc{Flip}
  call, we show that the calls to \textsc{FlipEdges}$(\A,r)$ for all
  maximal roots $r$ can be performed in time $O(kn)$.

  We first examine the complexity of the auxiliary step
  \textsc{FindEdges}$(r)$ for a given maximal root $r$.  This
  procedure requires a scan of the nodes of trees $T_s$ rooted at the
  maximal children $s$ of $r$ together with their outgoing edges.
  Since the edges contained in the lists $L_s$ have distinct target
  states in $T$, the flips of edges in $L_s$ can be performed with a
  time complexity at most $E(r)$, where $E(r)$ is the number of edges
  going out of or coming in a node of the tree $T$ rooted by $r$.
  Indeed, the update of the inverse structure for nodes in $T$ can be
  performed one time for all the flips of edges in $L_s$.  Note that
  the updating of the data after the flips is at most the size
  of~$T$. Indeed, after a flip of $(t,b,p)$ and $(t,a,p')$ only the
  nodes belonging to trees rooted at nodes along the red path from $p$
  to $t$ are updated.  As a consequence, the contribution of the
  auxiliary step in \textsc{FindStablePair} is $O(\sum_r E(r))=
  O(2kn)$.

  We now come to the complexity induced by the main treatment. We
  denote by $\Sect(r)$ the number of edges coming in or going out of a
  node belonging to the sector $J(r)$, \ie the nodes contained in a
  tree attached to a node of the cycle between $r'$ and $r$ ($r$
  included and $r'$), where $r'$ is the maximal root preceding $r$ on
  $C$. Let us compute for instance the complexity of the procedure
  \textsc{UniqueChildFlipEdges}$(\A,r,e=(t_1,b_1,p_1))$ (see
  Section~\ref{section.pseudocode}).  It contains at most two flips of
  edges ending in $T$.  The height of the tree $T_0$ is easily
  computed by scanning all nodes attached to some node of $C$ between
  $r$ and $r'$ ($r$ and $r'$ both excluded). In the case where this
  height is equal to $\ell$ and the set of outgoing edges of $s_0$ is
  a bunch, we flip the edge $e$.  We perform the procedure
  \textsc{UpDateSector}$(r, e)$ for updating the data of the nodes
  contained in the trees whose roots belong to $J(r)$.  Then we call
  a second (and last) time \textsc{FlipEdges}$(\A,r)$.  Since the time
  complexity of \textsc{UpDateSector}$(r, e)$ is at most $\Sect(r)$,
  we get that the time complexity of
  \textsc{UniqueChildFlipEdges}$(\A,r,e)$ is also
  $\Sect(r)$. Similarly, the time complexity of the procedures
  \textsc{ChildrenFlipEdgesEqual} and
  \textsc{ChildrenFlipEdgesUnEqual} is also $\Sect(r)$.

Hence the overall time spent for computing \textsc{FlipEdges}$(\A,r)$
for all maximal roots $r$ is $O(\sum_r \Sect(r))= O(2kn)$.
  The space complexity is $O(kn)$. Indeed, only linear
  additional space is needed to perform all operations.
\end{proof}

\section{The case of periodic graphs} \label{section.periodic}

Recall that the period of an automaton is the gcd of the lengths of
its cycles.  If the automaton $\A$ is an $n$-state
complete deterministic irreducible automaton which is not aperiodic,
it is not equivalent to a synchronized automaton. Nevertheless, the
previous algorithm can be modified as follows for finding an equivalent
automaton with the minimal possible rank. It has a quadratic-time
complexity.

\begin{small}
\begin{codebox}
\Procname{$\proc{PeriodicFindColoring }(\text{automaton } \A)$}
\li  $\B \gets \A$
\li  \While (size($\B$) $ > 1$) 
\li      \Do \proc{Update($\B$)}
\li       $\B, (s,t) \gets$ \proc{FindStablePair($\B$)}
\label{li:PeriodicFindColoring-3}
\li       lift the coloring up from $\B$ to the automaton $\A$ 
\li       \If there is a stable pair $(s,t)$ 
\label{li:PeriodicFindColoring-4}
\li              \Then $\B \gets$ \proc{Merge$(\B,(s,t))$}
\li              \Else \Return $\A$
          \End
     \End
\li \Return $\A$ 
\end{codebox}
\end{small}
It may happen that \textsc{FindStablePair} returns an automaton $\B$
which has no stable pair (it is made of a cycle where the set of
outgoing edges of any state is a bunch).  Lifting up this coloring
to the initial automaton $\A$ leads to a coloring of the initial
automaton whose minimal rank is equal to its period.

This result can be stated as the following theorem, which extends the Road
Coloring Theorem to the case of periodic graphs.

\begin{theorem}
  Any irreducible automaton $\A$ is equivalent to a an automaton whose
  minimal rank is the period of $\A$.
\end{theorem}
 
\begin{proof}

  Let us assume that $\A$ is equivalent to an automaton $\A'$ which
  has a stable pair $(s,t)$. Let $\B'$ be the quotient of $\A'$ by the
  congruence generated by $(s,t)$. Let $d$ be the period of $\A'$
  (equal to the period of $\A$) and $d'$ the period of $\B'$. 
  Let us show that $d = d'$.

It is clear that $d'$ divides $d$ (which we denote $d' / d$). Let
$\ell$ be the length of a path from $s$ to $s'$ in $\A'$, where $s'$
is equivalent to $s$. Since $(s,s')$ is stable, it is
synchronizable. Thus there is a word $w$ such that $s \cdot w = s'
\cdot w$. Since the automaton $\A'$ is irreducible, there is a path
labeled by some word $u$ from $s \cdot w$ to $s$. Hence $d / (\ell +
|w| + |u|)$ and $d / |(w| + |u|)$, implying $d / \ell$.  Let $\bar{s}$
be the class of $s$ and $z$ be the label of a cycle around $\bar{s}$
in $\B'$. Then there is a path in $\A'$ labeled by $z$ from $s$ to
$x$, where $x$ is equivalent to $x$. Thus $d / |z|$. It follows that
$d / d'$ and $d = d'$.

 Suppose that $\B'$ has rank $r$.  Let us show that $\A'$ also has
 rank $r$. Let $I$ be a minimal image of $\A'$ and $J$ be the set of
 classes of the states of $I$ in $\B'$. Two states of $I$ cannot
 belong to the same class since $I$ would not be minimal otherwise.
 As a consequence $I$ has the same cardinal as $J$. The set $J$ is a
 minimal image of $\B'$. Indeed, for any word $v$, the set $J \cdot v$
 is the set of classes of $I \cdot v$ which is a minimal image of
 $\A'$. Hence $|J \cdot v| = |J|$. As a consequence, $\B'$ has rank $r$.

  Let us now assume that $\A$ has no equivalent automaton which has a
  stable pair. In this case, we know that $\A$ is made of one red
  cycle where the set of edges going out of any state is a bunch. The
  rank of this automaton is equal to its period which is the
  length of the cycle. 

  Hence the procedure \textsc{PeriodicFindColoring} returns an
  automaton equivalent to $\A$ whose minimal rank is equal to its
  period.
\end{proof}

Since the modification of \textsc{FindSynchronizedColoring} into
\textsc{PeriodicFindColoring} does not change its complexity, we
obtain the following corollary.
\begin{corollary}
Procedure \textsc{PeriodicFindColoring} finds a coloring of minimal rank
for an $n$-state irreducible automaton in time $O(kn^2)$.
\end{corollary}

\section{Pseudocode} \label{section.pseudocode}
This section contains the pseudocode of some main procedures.
\subsection{Procedure \textsc{Merge}}

The computation of the congruence generated by $(s,t)$ can be
performed by using usual \textsc{Union/Find} functions computing
respectively the union of two classes and the leader of the class of a
state. After merging two classes whose leaders are $p$ and $q$
respectively, we need to merge the classes of $p \cdot
\ell$ and $q \cdot \ell$ for any $\ell \in A$.  
A pseudocode for merging classes is given in Procedure \textsc{Merge} below. 

\begin{small}
\begin{codebox}
\Procname{$\proc{Merge }(\text{automaton } \A, \text{ stable pair } (s,t))$}
\li  $x \gets \proc{Find}(s)$
\li  $y \gets \proc{Find}(t)$
\li  \If $x \neq y$ 
\li        \Then $\proc{Union}(x,y)$
\li               \For $\ell \in A$ 
\li                   \Do $\proc{Merge}(\A, (x \cdot \ell ,y \cdot \ell))$
                   \End
     \End
\li \Return $\A$
\end{codebox}
\end{small}

\subsection{Procedure \textsc{FlipEdges}}

We give below a pseudocode of the procedure
\textsc{FlipEdges}$(\A,r)$.  For each maximal root $r$, it
returns either an automaton equivalent to $\A$ together with a stable
pair, or an automaton equivalent to $\A$ together with one edge.
It performs some flips depending on the type of the edges returned by
\textsc{FindEdges}$(r)$. It calls \textsc{UniqueChildFlipEdges}$(r,e)$
in the case $r$ has a unique maximal child and $e$ is an edge of type 2
returned by \textsc{FindEdges}$(r)$.  It calls
\textsc{ChildrenFlipEdgesUnequal}$(\A,r)$ in the case $r$ has at least
two maximal children and \textsc{FindEdges}$(r)$ return a pair of
edges with distinct starting states.  It calls
\textsc{ChildrenFlipEdgesUnequal}$(\A,r)$ in the case $r$ has at least
two maximal children and \textsc{FindEdges}$(r)$ returns a pair of
edges which have the same starting state.

Recall that \textsc{GetPredecessor}$(r)$ returns the predecessor of
state $r$ on its red cycle.
\begin{small}
\begin{codebox}
\Procname{$\proc{FlipEdges}(\text{ automaton } \A, \text{ maximal root $r$})$}
\li     result $\gets \proc{FindEdges}(r)$ 
\li    \If ($r$ a unique maximal child $s_1$) and (result $\neq (3,e)$)
\li          \Then \If (result $=(0,e)$ or (result $=(2,e)$ where $e$ has type 1)
\li                     \Then \proc{Flip}$(e)$ 
\li                           \Return $\A$ and the stable pair $(s_1, \proc{GetPredecessor}(r))$
\li                     \Else (result $=(2,e)$ where $e$ has type 2)
\li                           \Return  $\proc{UniqueChildFlipEdges}(r,e)$ 
                    \End
        \End
\li     \If  ($r$ at least two maximal children) and (result $=(1,e_1,e_2)$ \\
         where $e_1 = (t_1,b_1,p_1),e_2= (t_2,b_2,p_2)$ have type 1 or 2)
\li          \Then  \If $t_1 \neq t_2$
\li                     \Then \Return  $\proc{ChildrenFlipEdgesUnequal}(r,e_1,e_2)$ 
\li                     \Else \Return  $\proc{ChildrenFlipEdgesEqual}(r,e_1,e_2)$ 
                     \End
         \End
\li      \If result $=(3,e)$ where $e$ is an edge of type 3
\li          \Then \Return $\A, e$
         \End 
\end{codebox}
\end{small}

\begin{small}
\begin{codebox}
\Procname{$\proc{UniqueChildFlipEdges }(\text{automaton } \A, 
         \text{ maximal root $r$}, \text{ edge } e=(t_1,b_1,p_1) \text{ of type 2})$}
\li      let $s_1$ be the unique child of $r$
\li      $s_0 \gets \proc{GetPredecessor}(r))$
\li       let $\T_0$ be the tree rooted at $r$ obtained by the potential flip of $e$ and the red edge 
\zi     going out of $t_1$, keeping only $r$ and the subtree rooted at the child $s_0$
\li        \If   $\height(\T_0) > \height(\T)$
\li             \Then \proc{Flip}$(t_1,b_1,p_1)$   
\li                   \Return $\A$ and the stable pair $(s_1,s_0)$ 
                \End
\li         \If  $\height(\T_0) < \height(\T)$
\li             \Then \Return $\A$ and the edge $e$
                 \End
\li         \If  $\height(\T_0) = \height(\T)$
\li             \Then \If the set of outgoing edges of $s_0$ and $s_1$ are bunches
\li                    \Then \Return $\A$ and the stable pair $(s_0,s_1)$ 
                       \End
\li                    \If the set of outgoing edges of $s_0$ is a bunch 
\zi                            and the set of outgoing edges of $s_1$ is not a bunch
\li                           \Then  $\proc{Flip}(t_1,b_1,p_1)$ 
\li                                   \proc{UpDateSector}$(r, e)$  (we still have $\height(\T_0) = \height(\T)$) 
\li                                   \Return $\proc{FlipEdges}(\A,r)$
                        \End
\li                    \If the set of outgoing edges of $s_0$ is not a bunch 
\li                          \Then let $(s_0,b,q_0)$ a $b$-edge going out of $s_0$ with $q_0 \neq r$
\li                                \If  $q_0 \notin \T$ 
\li                                    \Then $\proc{Flip}(s_0,b_0,q_0)$
\li                                         \If  the level of $q_0$ is positive
\li                                              \Then $r_0 \gets$ the root of the tree containing $q_0$
\li                                                     $s \gets \proc{GetPredecessor}(r_0)$
\li                                                     $t \gets$ the child of $r_0$ ancestor of $q_0$
\li                                                     \Return $\A$ and the stable pair $(s,t)$ 
\li                                              \Else $r_0 \gets$ the root of the tree containing $q_0$
\li                                                     $s \gets \proc{GetPredecessor}(r_0)$
\li                                                     \Return $\A$ and the stable pair $(s,s_0)$ 
                                                  \End
\li                                    \Else ($q_0 \in \T$ and $q_0 \neq r$) 
\li                                           \Return $\A$ and the edge $(s_0,b,q_0)$
                                       \End 
                            \End
           \End
\end{codebox}
\end{small}

\begin{small}
\begin{codebox}
\Procname{$\proc{ChildrenFlipEdgesEqual }(\text{automaton } \A, 
         \text{ maximal root $r$}, \text{ edges } e_1,e_2) \text{ of type } 2$}
\li      set $e_1 = (t_1,b_1,p_1)$ and $e_2 = (t_1,b_2,p_2)$
\li      $s_0 \gets \proc{GetPredecessor}(r)$
\li       let $\T_0$ be the tree rooted at $r$ obtained obtained by the potential flip of $(t_1,b_1,p_1)$
\zi and the red edge going out of $t_1$, keeping only $r$ and the subtree rooted at $s_0$
\li        \If      $\height(\T_0) > \height(\T)$
\li             \Then \proc{Flip}$(t_1,b_1,p_1)$   
\li                   \Return $\A$ and the stable pair $(s_1,s_0)$ 
                \End
\li         \If  $\height(\T_0) < \height(\T)$
\li             \Then \proc{Flip}$(t_1,b_1,p_1)$ 
\li                   \proc{UpDateSector}$(r, e_1)$ 
\li                    \Return  \proc{FlipEdges}$(\A, r)$ 
                 \End
\li        \If  $\height(\T_0) = \height(\T)$
\li              \Then \If the set of outgoing edges of $s_0$ is a bunch and there is an 
\zi     integer $i\geq 1$ such that the set of outgoing edges of $s_i$ is a bunch
\li                       \Then \Return $\A$ and the stable pair $(s_0,s_i)$ 
                          \End
\li                     \If the set of outgoing edges of $s_0$ is a bunch 
\zi                            and the sets of outgoing edges of $s_i$ for $i \geq 1$ are not bunches
\label{li:FlipEdgesChildren-cluster4-debut} 
\li                        \Then  $\proc{Flip}(t_1,b_1,p_1)$ 
\li                                        \proc{UpDateSector}$(r, e_1)$ (we still have $\height(\T_0) = \height(\T)$) 
\label{li:FlipEdgesChildren-cluster4-fin}
\li                                         \Return $\proc{FlipEdges}(\A,r)$
                           \End
\li                     \If the set of outgoing edges of $s_0$ is not a bunch 
\li                          \Then let $(s_0,b,q_0)$ a $b$-edge going out of $s_0$ with $q_0 \neq r$
\li                                \If  $q_0 \notin \T$ 
\li                                    \Then $\proc{Flip}(s_0,b_0,q_0)$
\li                                         \If  the level of $q_0$ is positive
\li                                              \Then $r_0 \gets$ the root of the tree containing $q_0$
\li                                                     $s \gets \proc{GetPredecessor}(r_0)$
\li                                                     $t \gets$ the child of $r_0$ ancestor of $q_0$
\li                                                     \Return $\A$ and the stable pair $(s,t)$ 
\li                                              \Else $r_0 \gets$ the root of the tree containing $q_0$
\li                                                     $s \gets \proc{GetPredecessor}(r_0)$
\li                                                     \Return $\A$ and the stable pair $(s,s_0)$ 
                                                  \End
\li                                    \Else ($q_0 \in \T$) 
\li                                         \If $q_0$ is not a descendant of $s_1$
\li                                          \Then  $\proc{Flip}(t_1,b_1,q_1)$
\li                                                 $\proc{Flip}(s_0,b_0,q_0)$    
\li                                                 $t \gets$ the child of $r$ ancestor of $q_0$
\li                                                 \Return $\A$ and the stable pair $(s_1,t_1)$ 
\li                                           \Else ($q_0$ is a descendant of $s_1$)
\li                                                 $\proc{Flip}(t_2,b_2,q_2)$
\li                                                 $\proc{Flip}(s_0,b_0,q_0)$    
\li                                                 \Return $\A$ and the stable pair $(s_1,s_2)$ 
                                              \End
                                       \End 
                            \End
           \End
\end{codebox}
\end{small}

The procedure \textsc{UpDateSector}$(r, e=(t_1,b_1,p_1))$ is called
after a flip of the edge $e$ and the red edge going out of $t_1$. It
updates the data of the nodes (and their trees attached to) along the
red path going from $p_1$ to $s_1$, where $s_1$ is the unique maximal
child of $r$.

\begin{small}
\begin{codebox}
\Procname{$\proc{ChildrenFlipEdgesUnequal }(\text{automaton } \A, 
         \text{ maximal root $r$}, \text{ edges } e_1,e_2)$}
\li  set $e_1 = (t_1,b_1,p_1)$ and $e_2 = (t_2,b_2,p_2)$ with $t_1 \neq t_2$
\li  \If at least one of $e_1, e_2$ (say $e_1$) has type 1 and $s_1$ is the child of $r$ ancestor of $p_1$
\li      \Then $s_0 \gets \proc{GetPredecessor}(r)$
\li            \proc{Flip}$(t_1,b_1,p_1)$ 
\li            \Return $\A$ and the stable pair $(s_0,s_1)$
\li       \Else \proc{Flip}$(t_1,b_1,p_1)$ 
\li             let $\T'$ be the new tree rooted at $r$
\li             \If   $\height(\T') > \height(\T)$ 
\label{li:FlipEdgesChild-calcul}
\li                        \Then \Return the stable pair $(s_1,s_0)$
\li                        \Else  $(\height(\T') \leq \height(\T)$)
\li                              \proc{Flip}$(t_2,b_2,p_2)$   
\li                              \Return $\A$ and the stable pair $(s_1,s_2)$
                 \End
          \End
      \End
\end{codebox}
\end{small}

\noindent \emph{Acknowledgments} The authors would like to thank
Florian Sikora, Avraham Trahtman, and the anonymous referees for 
pointing us some missing configurations in the algorithm.
We also thank the referees for helping us to improve the presentation
of the paper.

\bibliographystyle{abbrv}
\bibliography{road}

\begin{thebibliography}{10}

\bibitem{AdlerCoppersmithHassner83}
R.~L. Adler, D.~Coppersmith, and M.~Hassner.
\newblock Algorithms for sliding block codes.
\newblock {\em IEEE Trans. Inform. Theory}, IT-29:5--22, 1983.

\bibitem{AdlerGoodwynWeiss77}
R.~L. Adler, L.~W. Goodwyn, and B.~Weiss.
\newblock Equivalence of topological {M}arkov shifts.
\newblock {\em Israel J. Math.}, 27(1):48--63, 1977.

\bibitem{BealPerrin2008}
M.-P. B{\'e}al and D.~Perrin.
\newblock A quadratic algorithm for road coloring.
\newblock {\em CoRR}, abs/0803.0726, 2008.

\bibitem{BerstelPerrinReutenauer2010}
J.~Berstel, D.~Perrin, and C.~Reutenauer.
\newblock {\em Codes and automata}, volume 129 of {\em Encyclopedia of
  Mathematics and its Applications}.
\newblock Cambridge University Press, Cambridge, 2010.

\bibitem{BoyleMass2000}
M.~Boyle and A.~Maass.
\newblock Expansive invertible onesided cellular automata.
\newblock {\em J. Math. Soc. Japan}, 52(4):725--740, 2000.

\bibitem{BoyleMass2004}
M.~Boyle and A.~Maass.
\newblock Erratum to: "{E}xpansive invertible onesided cellular automata" [{J}.
  {M}ath. {S}oc. {J}apan {\bf 52} (2000), no. 4, 725--740].
\newblock {\em J. Math. Soc. Japan}, 56(1):309--310, 2004.

\bibitem{BudzbanFeinsilver11}
G.~Budzban and P.~Feinsilver.
\newblock The generalized road coloring problem and periodic digraphs.
\newblock {\em Appl. Algebra Eng. Commun. Comput.}, 22(1):21--35, 2011.

\bibitem{CarayolNicaud2012}
A.~Carayol and C.~Nicaud.
\newblock Distribution of the number of accessible states in a random
  deterministic automaton.
\newblock In {\em STACS}, volume~14 of {\em LIPIcs}, pages 194--205. Schloss
  Dagstuhl - Leibniz-Zentrum fuer Informatik, 2012.

\bibitem{Carbone01}
A.~Carbone.
\newblock Cycles of relatively prime length and the road coloring problem.
\newblock {\em Israel J. Math.}, 123:303--316, 2001.

\bibitem{CulikEtAl1999}
K.~Culik, II, J.~Karhum{\"a}ki, and J.~Kari.
\newblock Synchronized automata and road coloring problem.
\newblock Technical report, TUCS Technical Report 323, Turku Center for
  Computer Science, University of Turku, 1999.

\bibitem{CulikEtAl02}
K.~Culik, II, J.~Karhum{\"a}ki, and J.~Kari.
\newblock A note on synchronized automata and road coloring problem.
\newblock In {\em Developments in language theory (Vienna, 2001)}, volume 2295
  of {\em Lecture Notes in Comput. Sci.}, pages 175--185. Springer, Berlin,
  2002.

\bibitem{DelyonMaler94}
B.~Delyon and O.~Maler.
\newblock On the effects of noise and speed on computations.
\newblock {\em Theoret. Comput. Sci.}, 129(2):279--291, 1994.

\bibitem{Eilenberg76B}
S.~Eilenberg.
\newblock {\em Automata, Languages, and Machines. {V}ol. {B}}.
\newblock Academic Press, New York, 1976.

\bibitem{Eppstein1990}
D.~Eppstein.
\newblock Reset sequences for monotonic automata.
\newblock {\em SIAM J. Comput.}, 19(3):500--510, 1990.

\bibitem{FreilingEtAl2003}
C.~F. Freiling, D.~S. Jungreis, F.~Th{\'e}berge, and K.~Zeger.
\newblock Almost all complete binary prefix codes have a self-synchronizing
  string.
\newblock {\em IEEE Transactions on Information Theory}, 49(9):2219--2225,
  2003.

\bibitem{Friedman90}
J.~Friedman.
\newblock On the road coloring problem.
\newblock {\em Proc. Amer. Math. Soc.}, 110(4):1133--1135, 1990.

\bibitem{JonoskaKarl1999}
N.~Jonoska and S.~A. Karl.
\newblock A molecular computation of the road coloring problem.
\newblock In {\em D{NA} based computers, {II} ({P}rinceton, {NJ}, 1996)},
  volume~44 of {\em DIMACS Ser. Discrete Math. Theoret. Comput. Sci.}, pages
  87--96. Amer. Math. Soc., Providence, RI, 1999.

\bibitem{Jurgensen2008}
H.~J{\"u}rgensen.
\newblock Synchronization.
\newblock {\em Inform. and Comput.}, 206(9-10):1033--1044, 2008.

\bibitem{Kari03}
J.~Kari.
\newblock Synchronizing finite automata on {E}ulerian digraphs.
\newblock {\em Theoret. Comput. Sci.}, 295(1-3):223--232, 2003.

\bibitem{KariVolkov2013}
J.~Kari and M.~V. Volkov.
\newblock {{\v C}ern{\'y}}'s conjecture and the road coloring problem.
\newblock In {\em Handbook of Automata}. European Science Foundation, 2013.
\newblock to appear.

\bibitem{LindMarcus95}
D.~A. Lind and B.~H. Marcus.
\newblock {\em An Introduction to Symbolic Dynamics and Coding}.
\newblock Cambridge University Press, Cambridge, 1995.

\bibitem{Nicaud2013}
C.~Nicaud.
\newblock On the synchronization of random deterministic automata.
\newblock preprint, 2013.

\bibitem{OBrien81}
G.~L. O'Brien.
\newblock The road-colouring problem.
\newblock {\em Israel J. Math.}, 39(1-2):145--154, 1981.

\bibitem{OlschewskiUmmels2010}
J.~Olschewski and M.~Ummels.
\newblock The complexity of finding reset words in finite automata.
\newblock In {\em MFCS}, volume 6281 of {\em Lecture Notes in Computer
  Science}, pages 568--579. Springer, 2010.

\bibitem{PerrinSchutzenberger92}
D.~Perrin and M.-P. Sch{\"u}tzenberger.
\newblock Synchronizing prefix codes and automata and the road coloring
  problem.
\newblock In {\em Symbolic dynamics and its applications}, volume 135 of {\em
  Contemp. Math.}, pages 295--318. Amer. Math. Soc., 1992.

\bibitem{PomeranzReddy1994}
I.~Pomeranz and S.~M. Reddy.
\newblock Application of homing sequences to synchronous sequential circuit
  testing.
\newblock {\em IEEE Trans. Computers}, 43(5):569--580, 1994.

\bibitem{Roman2011}
A.~Roman.
\newblock The {NP}-completeness of the road coloring problem.
\newblock {\em Inform. Process. Lett.}, 111(7):342--347, 2011.

\bibitem{SkvortsovZaks2010}
E.~S. Skvortsov and Y.~Zaks.
\newblock Synchronizing random automata.
\newblock {\em Discrete Mathematics {\&} Theoretical Computer Science},
  12(4):95--108, 2010.

\bibitem{Trahtman09}
A.~N. Trahtman.
\newblock The road coloring problem.
\newblock {\em Israel J. Math.}, 172:51--60, 2009.

\bibitem{Trahtman2010}
A.~N. Trahtman.
\newblock A partially synchronizing coloring.
\newblock In {\em Proceedings of CSR 2010}, volume 6072 of {\em Lecture Notes
  in Comput. Sci.}, pages 362--370. Springer-Verlag, 2010.

\bibitem{Trahtman2011}
A.~N. Trahtman.
\newblock An algorithm for road coloring.
\newblock In {\em Combinatorial algorithms}, volume 7056 of {\em Lecture Notes
  in Comput. Sci.}, pages 349--360. Springer, Heidelberg, 2011.

\end{thebibliography}


\end{document}